\documentclass[notitlepage,a4paper,aps,prd,onecolumn,superscriptaddress,nofootinbib,groupedaddress,longbibliography]{revtex4-1}
\usepackage{amsmath}
\usepackage{amsfonts}
\usepackage{amssymb}
\usepackage[utf8]{inputenc}
\usepackage[T1]{fontenc}
\usepackage[colorlinks=true]{hyperref}
\usepackage[dvipsnames]{xcolor}
\usepackage{graphicx}
\usepackage[normalem]{ulem}
\usepackage{enumerate}
\usepackage{pdfpages}

\makeatletter
\AtBeginDocument{\let\LS@rot\@undefined}
\makeatother

\newtheorem{theorem}{Theorem}

\newtheorem{definition}[theorem]{Definition}

\newtheorem{lemma}[theorem]{Lemma}

\newtheorem{proposition}[theorem]{Proposition}

\newenvironment{proof}[1][Proof]{\noindent\textbf{#1.} }{\ \rule{0.5em}{0.5em}}
\begin{document}
\title{Finsler gravity action from variational completion}

\author{Manuel Hohmann}
\email{manuel.hohmann@ut.ee}
\affiliation{Laboratory of Theoretical Physics, Institute of Physics, University of Tartu, W. Ostwaldi 1, 50411 Tartu, Estonia}

\author{Christian Pfeifer}
\email{christian.pfeifer@ut.ee}
\affiliation{Laboratory of Theoretical Physics, Institute of Physics, University of Tartu, W. Ostwaldi 1, 50411 Tartu, Estonia}

\author{Nicoleta Voicu}
\email{nico.voicu@unitbv.ro}
\affiliation{Faculty of Mathematics and Computer Science, Transilvania University, Iuliu Maniu Str. 50, 500091 Brasov, Romania}

\begin{abstract}
In the attempts to apply Finsler geometry to construct an extension of general relativity, the question about a suitable generalization of the Einstein equations is still under debate. Since Finsler geometry is based on a scalar function on the tangent bundle, the field equation which determines this function should also be a scalar equation. In the literature two such equations have been suggested: the one by Rutz and the one by one of the authors. Here we employ the method of canonical variational completion to show that Rutz equation can not be obtained from a variation of an action and that its variational completion yields the latter field equations. Moreover, to improve the mathematical rigor in the derivation of the Finsler gravity field equation, we formulate the Finsler gravity action on the positive projective tangent bundle. This has the advantage of allowing us to apply the classical variational principle, by choosing the domains of integration to be compact and independent of the dynamical variable. In particular in the pseudo-Riemannian case, the vacuum field equation becomes equivalent to the vanishing of the Ricci tensor.
\end{abstract}

\maketitle

\section{Introduction}\label{sec:intro}
A large source of information about the physical properties of spacetime is obtained by observing the motion of point particles. The observed trajectories are identified with the geodesics of the geometry of spacetime and thus, by matching the observed curves with the predicted geodesics, the viability of a certain geometry can be tested. Conversely, when a spacetime geometry is determined by dynamical physical field equations, its geodesics can be derived and observable effects can be predicted.

General relativity is based on pseudo-Riemannian geometry, i.e.\ a spacetime manifold equipped with a metric tensor of Lorentzian signature. The metric is determined by the Einstein equations, its geodesics predict the motion of point particles and geodesic deviation is sourced by the curvature of its Levi-Civita connection. On a huge variety of physical scales the predictions made on the basis of general relativity are outstandingly correct and in agreement with observation, however there are the well known shortcomings such as the rotational curves of galaxies and the accelerated expansion of the universe, which led to the introduction of the notions of dark matter and dark energy~\cite{Riess:1998cb,Peebles:2002gy,Corbelli:1999af,Clowe:2006eq}. The most common approach to understand and explain dark matter and dark energy is to postulate the existence of additional particles to the ones in the standard model of particle physics, and the alternative is to look for extensions and modifications in the description of gravity \cite{Papantonopoulos:2007zz}. In this article we follow the latter route and consider Finsler geometry as extended geometry of spacetime, which has been proposed as one possibility to shed light onto the dark universe phenomenology \cite{Hohmann:2016pyt,Basilakos:2013hua,Papagiannopoulos:2017whb,Chang:2009pa,Basilakos:2017rgc}.

Finsler spacetime geometry is the geometry of a manifold equipped with a so called Finsler function, which is a $1$-homogeneous function on the tangent bundle of spacetime and defines a length measure for curves. It thus is the most general geometry with a geometric clock in the sense of the clock postulate, namely that the time an observer measures between two events is given by the length of its worldline.\\
Finsler geometry, respectively the geodesics of a Finslerian geometry, describe the motion of point particles subject to a dispersion relation, which is non-quadratic in the particle's four momenta. Such modified dispersion relations (MDRs) appear most naturally in effective field theories in media. Examples are premetric electrodynamics \cite{Hehl} describing among other systems the electromagnetic field in crystals \cite{Perlick} or wave equations in solids \cite{Cerveny} which can be used to model earthquake waves \cite{Seismic1}. Moreover MDRs are used as an effective description of quantum gravity effects \cite{Amelino-Camelia:2014rga,Lobo:2016xzq,Letizia:2016lew,Barcaroli:2017gvg}, making spacetime effectively a medium whose origin lies in the four momentum dependent scattering of elementary particles with the graviton. More fundamentally MDRs emerge from field theories which break local Lorentz invariance \cite{Raetzel:2010je}, as studied in the standard model extension \cite{Bluhm:2005uj}, very special and very general relativity \cite{Cohen:2006ky,Gibbons:2007iu,Fuster:2018djw} or again in premetric resp. area metric electrodynamics \cite{Punzi:2007di,Punzi:2009}.

The appearance of Finsler geometry in physics, see \cite{Pfeifer:2019wus} for a review, raised two questions: the one about a suitable general mathematical definition of Finsler spacetimes, which covers the interesting instances appearing in the literature and the one whether it is possible to find dynamical equations which determine the Finslerian spacetime geometry in the same way as the Einstein equations determine the pseudo-Riemannian geometry of spacetime.

Regarding the first issue there are several proposals in the literature, see for example \cite{Asanov,Pfeifer:2011tk,Lammerzahl:2012kw,Javaloyes:2018lex, Bejancu-Farran}, which extend Beem's original suggestion \cite{Beem}. However none of them captures the whole variety of Finsler spacetime geometries under investigation. As part of this work we suggest a definition of Finsler spacetimes which distills the necessary features from the existing definitions and includes large classes of the Finsler spacetime geometries discussed in the literature.

Regarding the dynamical equations determining a Finslerian spacetime geometry, also numerous attempts have been made \cite{Chang:2009pa,Kouretsis:2008ha,Stavrinos2014,Asanov,Voicu:2009wi,Bucataru,Rutz,Pfeifer,Vacaru:2010fi,Minguzzi:2014fxa}. The difference in the approaches lies in the choice of the fundamental variable: the Finsler function or the Finsler metric tensor (accordingly in the type of the dynamical equation, scalar or tensorial) and in the way how the equation is obtained, by variation from an action, by formal resemblance to the Einstein equations or from further physical principles.

In this article we argue for physical dynamics for Finsler spacetime geometries which have the following properties:
\begin{itemize}
	\item the fundamental variable is the geometry defining Finsler function, i.e.\ the field equation we are looking for is a homogeneous scalar equation on the tangent bundle, which determines a homogeneous function;
	\item the homogeneity of the Finsler function, which can be understood as equivariance with respect to certain group actions, allows us to naturally treat Finsler functions as sections of a certain fiber space sitting over a compact manifold: the positive projective tangent bundle. This way, we can apply the classical apparatus of the calculus of variations (see, e.g., \cite{Krupka-book});
	\item the field equation is obtained by variational means, starting from a well defined action integral;
	\item the geometrical structures used are as simple as possible and are constructed from the Finsler function alone;
	\item in the case when the Finslerian spacetime geometry is pseudo-Riemannian, the dynamics become equivalent to the dynamics determined by the Einstein equations
	\begin{align}
	G_{ij}=\dfrac{8\pi G}{c^{4}}T_{ij},  \label{full_einstein_eq}
	\end{align}
	where $G_{ij}=r_{ij}-\dfrac{1}{2}rg_{ij}$ are the components of the Einstein tensor built from the components of the Ricci tensor $r_{ij}$ and the Ricci scalar $r$ of the Levi-Civita connection of the Lorentzian spacetime metric with components~$g_{ij}$.\footnote{In order to distinguish Riemannian curvature-related geometric objects from Finslerian ones (as some of them have different definitions), we denote the curvature tensor in a pseudo-Riemannian space by $r$ (small letter) and the Finslerian curvature-related quantities by capital letters.}
\end{itemize}
We prove that the most promising conjectured Finsler spacetime dynamics, the one by Rutz \cite{Rutz} and the one by Pfeifer and Wohlfarth \cite{Pfeifer} (which was independently also found by Chen and Shen in the context of positive definite Finsler geometry \cite{Chen-Shen}), are actually related in the way that the latter is the \emph{variational completion} of the former. A similar property can be found in the emergence of the Einstein equations. An early version of the Einstein equations was simply stating that $r_{ij} \sim T_{ij}$. It has been shown that the left hand side of this equation cannot be obtained by variational calculus, not even in the vacuum case $T_{ij}=0$, and its variational completion is given by the Einstein tensor~\cite{Voicu-Krupka}. Hence by the demand of a variational equation for the Finsler function as fundamental variable of a Finslerian spacetime geometry, the simplest self consistent action based field equations are the ones which were derived in~\cite{Pfeifer}.

We establish this result in the following way. The variational completion algorithm is based on the notion of Vainberg-Tonti Lagrangian associated to a given system of partial differential equations. This Lagrangian is determined solely by the PDE system and in the case when the given PDE system is variational, it is the Lagrangian that admits this PDE system as its Euler-Lagrange equations. We find that the Vainberg-Tonti Lagrangian (regarded as a differential form) corresponding to Rutz's equation is the product between the (canonical, 0-homogeneous) Finslerian Ricci scalar and the canonical volume form built from the Finsler metric on the projective tangent bundle. Based on this Lagrangian we construct the action integral for the Finsler function on compact subsets of the positive projective tangent bundle, in order to correctly handle all technical difficulties which appear due to the homogeneity properties of the geometric objects involved. Variation of the obtained action with respect to the Finsler function, then yields the desired field equations. In the end it turns out that they are identical to the ones found in \cite{Pfeifer}, where the action was formulated on the unit tangent bundle defined by the Finsler function. We choose the positive projective tangent bundle as manifold on which the action is defined here, since, first, its fibers are compact (while the Finsler spacetime unit tangent bundle does not have this property), and, second, it avoids an intertwining between the definition of the integration domain and the fundamental dynamical variable, as it is present in the unit tangent bundle approach. Thus we add the missing puzzle pieces in mathematical rigor to action based dynamics for Finsler gravity and confirm the closest Finsler generalization of the Einstein--Hilbert action to be defined by the $0$-homogenized canonical Finsler curvature scalar.

The article is structured as follows. We begin by giving the basic definition of Finsler spacetimes and their geometry in Section \ref{sec:FST}. In Section \ref{sec:FSTdyn} we review the Finsler gravity field equation conjectured by Rutz and the one developed by Pfeifer and Wohlfarth. Afterwards in Section \ref{sec:FGraviAction} we introduce the positive projective tangent bundle as the stage where we formulate the Finsler gravity action. The main result is then presented in Section \ref{sec:RutzVT} where we show that the variational completion of Rutz equation is given by the one developed Pfeifer and Wohlfarth. We confirm the field equation by variational calculus in Section \ref{sec:vari}, before we show its consistency with the Einstein equations in the case that the Finsler spacetime geometry is a pseudo-Riemannian geometry in Section \ref{sec:metric}. Finally Section \ref{sec:conc} is devoted to conclusions. The main part of the article is supplemented by several appendices. In Appendix \ref{app:Conv} we prove the convexity of the set of timelike vectors on Finsler spacetimes. Two particular classes of Finsler spacetimes are discussed in Appendix~\ref{app:Ex}. Two lengthy integrals are evaluated in Appendix~\ref{app:I2I3}. The proof of Lemma \ref{lem:intbdry} is deferred to Appendix~\ref{app:Lemmaint}.

\section{Finsler spacetimes}\label{sec:FST}

We begin by stating the basic notations and definitions of Finsler geometry we use throughout this article. In section~\ref{ssec:fstdef} we provide the definition of Finsler spacetimes we employ in this article. We review their geometry in section~\ref{ssec:fstgeom}.

\subsection{The Definition}\label{ssec:fstdef}
Let $M$ be a connected, oriented, $\mathcal{C}^{\mathcal{\infty }}$-smooth manifold of dimension $4$ and $(TM,\pi_{TM},M)$, its tangent bundle. Let $\left\{ (U_{\alpha },\varphi _{\alpha })\right\}$ be an oriented atlas on $M$. We denote by $(x^{i})_{i=\overline{0,3}}$ the coordinates of a point $x\in M$ in a local chart $(U_{\alpha },\varphi _{\alpha })$; denoting, for any vector $\dot{x}\in T_{x}M$, by $(\dot{x}^{i})$ the coordinates in the local natural basis $\{\partial_i=\partial /\partial x^{i}\}$ of $T_{x}M$, we obtain, for a point $(x,\dot{x})\in \pi _{TM}^{-1}(U_{\alpha })\subset TM$, the coordinates $(x^{i},\dot{x}^{i})_{i=\overline{0,3}};$ then, $\left\{ (\pi_{TM}^{-1}(U_{\alpha }),d\varphi _{\alpha })\right\} $ is an oriented atlas on $TM$. We will denote by $_{,i}$ and $_{\cdot i}$ partial differentiation with respect with $x^{i}$ and $\dot{x}^{i}$ respectively.

By $\mathcal{F}(TM)$, we will mean the set of $\mathcal{C}^{\infty }$-smooth functions on $TM$. For any fibered manifold \(\pi: Y \to X\), we will denote by $\Gamma (Y)$ the module of sections of $Y$ and by $\Omega (Y),$ the set of differential forms on $Y$.

A \textit{conic subbundle} of $TM$ is, \cite{Javaloyes:2018lex}, a non-empty open submanifold $\mathcal{Q}\subset TM\backslash \{0\}$, with the following properties:
\begin{itemize}
	\item $\pi _{TM}(\mathcal{Q})=M$;
	\item \textit{conic property:} if $(x,\dot{x})\in \mathcal{Q}$, then, for any $\lambda >0:$ $(x,\lambda \dot{x})\in \mathcal{Q}$.
\end{itemize}
The first condition above ensures that $(\mathcal{Q},\pi _{TM|\mathcal{Q}},M)$ has a fibered manifold structure.

We formulate and employ a definition of Finsler spacetimes which is distilled from previous generalizations \cite{Pfeifer:2011tk,Lammerzahl:2012kw,Javaloyes:2018lex} of Beem's original definition \cite{Beem}, to include a most complete variety of indefinite Finsler length measures discussed in the literature.
\begin{definition}\label{def:FST}
By a \textit{Finsler spacetime}, we understand a pair $(M,L)$, where $L:TM\rightarrow \mathbb{R}$ is a continuous function, called the Finsler-Lagrange function, which satisfies:
	\begin{itemize}
		\item $L$ is positively homogeneous of degree two with respect to $y$: $L(x,\lambda \dot x) = \lambda^2 L(x,\dot x)$;
		\item $L$ is smooth and the vertical Hessian of $L$ (called $L$-metric $g^L$)
		\begin{align}\label{g_def}
			g^L_{ij}=\dfrac{1}{2}\dfrac{\partial ^{2}L}{\partial \dot{x}^{i}\partial \dot{x}^{j}} = \dfrac{1}{2} L_{\cdot i\cdot j}
		\end{align}
		is non-degenerate on a conic subbundle $\mathcal{A}$ of $TM$ such that  $TM\setminus \mathcal{A}$ is of measure zero;
		\item there exists a connected component $\mathcal{T}$ of the preimage $L^{-1}((0,\infty))\subset TM$, such that on $\mathcal{T}$ the $L$-metric $g^L$ exists, is smooth and has Lorentzian signature $(+,-,-,-)~$\footnote{It is possible to equivalently formulate this property with opposite sign of $L$ and metric $g^L$ of signature $(-,+,+,+)$. We fixed the signature and sign of $L$ here to simplify the discussion.}
		\item the Euler-Lagrange equations
		\begin{align}\label{eq:EL}
		\frac{d}{d \tau} \dot{\partial}_i L - \partial_i L = 0\,.
		\end{align}
		have a unique local solution for every initial condition $(x,\dot x)\in \mathcal{T}\cup \mathcal{N}$, where $\mathcal{N}$ is the kernel of $L$. At points of $\mathcal{N}$ where the $L$-metric degenerates the solution must be constructed by continuous extension. This means that the geodesic equation coefficients admit a $\mathcal{C}^{1}$ extension at those points.
	\end{itemize}
\end{definition}
The difficulty in the definition of Finsler spacetimes emerges from the existence of four conic subbundles of $TM\setminus\{0\}$ which characterize the properties of the indefinite Finsler geometry:
\begin{itemize}
	\item $\mathcal{A}$: the subbundle where $L$ is smooth and $g^L$ is non-degenerate, with fiber $\mathcal{A}_{x} = \mathcal{A} \cap T_xM$, called the set of \emph{admissible vectors},
	\item $\mathcal{N}$: the subbundle where $L$ is zero, with fiber $\mathcal{N}_x = \mathcal{N} \cap T_xM$,
	\item $\mathcal{A}_0 = \mathcal{A}\setminus\mathcal{N}$: the subbundle where $L$ can be used for normalization, with fiber  $\mathcal{A}_{0x} = \mathcal{A}_0 \cap T_xM$,
	\item $\mathcal{T}$: a maximally connected conic subbundle where $L > 0$, the $L$-metric exists and has Lorentzian signature $(+,-,-,-)$, with fiber $\mathcal{T}_x = \mathcal{T} \cap T_xM$.
\end{itemize}
The most important part of the definition is the existence of the subbundle $\mathcal{T}$ which ensures the existence of a convex cone $\mathcal{T}_x$ in each tangent space $T_xM$. A detailed proof of the convexity of $\mathcal{T}_x$ can be found in Appendix \ref{app:Conv}. The set $\mathcal{T}_x$ can be interpreted as set of future pointing timelike directions, which are allowed as tangent vectors to the trajectories of physical observers; see~\cite[Sec. V]{Pfeifer} for illustrative examples.

The relation of the subbundles $\mathcal{A}, \mathcal{N}, \mathcal{A}_0$ and $\mathcal{T}$ is basically what distinguishes the earlier definitions of Finsler spacetimes \cite{Pfeifer:2011tk,Lammerzahl:2012kw,Javaloyes:2018lex} and our new one. Here we obviously have that $\mathcal{T}\subset \mathcal{A}_0$ and $\mathcal{A}_0 \subset \mathcal{A}$, but require nothing else. Thus, in particular $L$ is smooth and $g^L$ is non-degenerate on all of $\mathcal{T}$, which ensures the existence of all geometric objects, introduced in the next section, on $\mathcal{T}$. We do not demand any relation between $\mathcal{N}$ and $\mathcal{A}$, hence $L$ may be not differentiable both along directions where $L(x,\dot x)\neq 0$ and where $L(x,\dot x)=0$.

Our definition includes large classes of Finsler spacetimes according to the older definitions.\footnote{The ones considered in \cite{Caponio:2017lgy,Caponio:2015hca} do not fit in our definition. We do not consider these Finsler spacetimes since for them the curvature tensor, which defines the dynamics of Finsler spacetimes, is not necessarily defined for all physical observer directions, which in our definition is given by the conic subbundle $\mathcal{T}$. The definition could be relaxed so as to include the possibility of having an observer direction where curvature is not defined, but in this case, a thorough analysis of whether the evolution of spacetime is causal, as seen by the respective observer, is needed. This is the subject for future work.}. It allows for example Finsler spacetime geometries of
\begin{itemize}
	\item Randers type $F = \sqrt{|g_{ab}\dot x^a \dot x^b|} + A_c\dot x^c$,
	\item Bogoslovsky/Kropina type $F = (|g_{ab}\dot x^a \dot x^b|)^{\frac{1-q}{2}}(A_c(x)\dot x^c)^{q}$,
	\item polynomial $m$-th root type $F = |G_{a_1 \cdots a_m}(x)\dot x^{a_1}\ldots\dot x^{a_m}|^{\frac{1}{m}}$.
\end{itemize}
Surely, the tensor fields $g$, $A$ and $G$ can not be arbitrary to obtain a viable Finsler spacetime. In Appendix~\ref{app:Ex} we briefly discuss conditions they have to satisfy so that the resulting $L$ fits into our definition. From the viewpoint of physics the Randers class describes the motion of a charged particle in an electromagnetic potential, the Bogoslovsky class is considered under the term very special or very general relativity~\cite{Fuster:2018djw,Gibbons:2007iu,Cohen:2006ky,Bogoslovsky} and the polynomial class, for $m=4$, describes the propagation of light in premetric electrodynamics \cite{Hehl,Rubilar:2007qm,Punzi:2007di,Gurlebeck:2018nme}. An extensive discussion on examples of our definition of Finsler spacetimes will be the topic of a forthcoming paper, and is not the main subject of this work, which deals with dynamical equations for Finsler spacetime geometries.

The Finsler function $F$, usually employed in standard textbooks about Finsler geometry \cite{Shen,Bucataru}, is defined as $F = \sqrt{|L|}$ and the length measure for curves $\gamma:[a,b]\mapsto M$ on $M$ is given by
\begin{align}\label{eq:length}
	\ell[\gamma] = \int_a^b F(\gamma,\dot{\gamma})\ d\tau\,.
\end{align}

\subsection{The Geometry}\label{ssec:fstgeom}
The geometry of Finsler spacetimes is constructed from objects obtained from derivatives acting on $L$. All details on geometry based on non-linear connections and of Finsler spacetimes can be found in the books \cite{Giachetta-book,Bucataru,Chern-Shen-Lam}. Here we recall the notions we need throughout this article. On the basis of our definition of Finsler spacetimes all objects are defined on the bundle $\mathcal{A}$ and not on all of $TM$.

In any local chart of $\mathcal{A}$ the first derivative of $L$ w.r.t. $\dot x$ defines the momenta, or lower index velocities,
\begin{align}
	p_{(x,\dot x)} = \dot x_i dx^i,\quad \dot x_i = \frac{1}{2}L_{\cdot i}\,,
\end{align}
the second derivatives of $L$ define the $L$-metric and its inverse
\begin{align}
	g^L_{(x,\dot x)} = g^L_{ij}(x,\dot x) dx^i \otimes dx^j,\quad g^L_{ij} = \frac{1}{2}L_{\cdot i\cdot j}\,,
\end{align}
and the third derivatives, the so called \emph{Cartan tensor}
\begin{align}
	C_{(x,\dot x)} = C_{ijk}(x,\dot x) dx^i \otimes dx^j \otimes dx^k,\quad C_{ijk} = \frac{1}{2}g^L{}_{ij\cdot k} = \frac{1}{4} L_{\cdot i \cdot j \cdot k}\,.
\end{align}
By the homogeneity of $L$ the following equalities hold in every local coordinate chart on $\mathcal{A}$
\begin{align}
	L(x,\dot x) = g^{L}_{ij}(x,\dot x) \dot x^i \dot x^j,\quad L_{\cdot i}(x.\dot x) = 2 \dot x_i = 2 g^{L}_{ij}(x,\dot x) \dot x^j,\quad \dot x_i{}_{\cdot j} = g^{L}_{ij}(x,\dot x),\quad \dot x^i C_{ijk}(x,\dot x)=0\,.
\end{align}

The fundamental ingredient of the geometry of a Finsler spacetime is the geodesic spray, from which one obtains the canonical non-linear connection, defining parallel transport. The geodesic equation of \eqref{eq:length} in arclength parametrization can be written as
\begin{align}\label{eq:geod}
	\ddot \gamma^i + 2 G^i(\gamma,\dot{\gamma}) = 0\,,
\end{align}
where the geodesic spray coefficients are given by
\begin{align}\label{eq:geodspray}
	2G^i = \frac{1}{2}g^{L ij}(\dot x^k L_{,k\cdot j} - L_{,j})\,.
\end{align}
They define the coefficients $G^i{}_j$ of the canonical \emph{Cartan non-linear connection}, which will be understood as defining a splitting of the tangent bundle $(T\mathcal{A},\pi_{\mathcal{A}},\mathcal{A})$ of $\mathcal{A}$ into a \emph{vertical subbundle} $V\mathcal A = \mathrm{ker}(d\pi_{|_{\mathcal{A}}})$ and a \emph{horizontal subbundle} $H\mathcal{A}$ such that $T\mathcal A = H\mathcal A \oplus V\mathcal A $. The local adapted basis will be denoted by $(\delta_i, \dot{\partial}_i)$, where $\delta_i = \partial_i - G^j{}_i \dot{\partial}_j$ and $\dot{\partial}_i = \partial_{\dot x^i}$. The connection coefficients are defined as
\begin{align}\label{eq:nldef}
	G^i{}_j = G^i{}_{\cdot j}\,.
\end{align}
Besides the fundamental non-linear connection it is possible to define several linear connections on Finsler spacetimes. We do not regard these linear connections as fundamental but rather as tools to define tensorial quantities. For the purposes of this article we will use the so called Chern-Rund linear connection $\mathrm D$ on $TM$ restricted to $\mathcal{A}$. It is locally defined by
\begin{align}\label{eq:cherndef}
	\mathrm D_{\delta_k}\delta_j = \Gamma^i{}_{jk}\delta_i,\quad \mathrm D_{\delta_k}\dot{\partial}_j = \Gamma^i{}_{jk}\dot{\partial}_i,\quad \mathrm D_{\dot{\partial}_k}\delta_j = \mathrm D_{\dot{\partial}_k}\dot{\partial}_j = 0\,,
\end{align}
where $\Gamma^{i}{}_{jk}:=\frac{1}{2}g^{Lih}(\delta _{k}g^L_{hj}+\delta_{j}g^L_{hk}-\delta _{h}g^L_{jk}).$ We denote by $_{|i}$ $\mathrm D$-covariant
differentiation with respect to $\delta _{i}$. The difference between the derivative of the non-linear connection coefficients $G^i{}_{jk} = G^i{}_{j\cdot k}$ and the Chern-Rund connection coefficients $\Gamma^i{}_{jk}$ defines the \emph{Landsberg tensor} $P=P^i{}_{jk} \delta_{i} \otimes d x^j \otimes d x^k$, with
\begin{align}\label{eq:landsberg}
P^i{}_{jk} = G^i{}_{jk} -\Gamma^i{}_{jk},\quad \dot x^j P^i{}_{jk}(x,\dot x) = 0\,.
\end{align}

The geometric objects introduced so far satisfy some important identities regarding their differentiation w.r.t the Chern Rund connection and the dynamical covariant derivative $\nabla:\Gamma(T\mathcal{A})\rightarrow\Gamma(T\mathcal{A})$, which is attached to the non-linear connection~\cite{Bucataru}:
\begin{align}
	\delta_i L =  L_{|i} &= 0,\quad \nabla L = 0\\
	\dot x^i_{|j} &= 0,\quad \nabla \dot x^i = 0\\
	g^L_{ij|k} &= 0,\quad \nabla g^L_{ij} = 0\\
	\nabla C^i{}_{jk} &= P^i{}_{jk}\,.
\end{align}
The derivative operators are related by the identity $\nabla = \dot x^i \mathrm D_{\delta_i}$.

To understand the motivation of the Finsler gravity equation suggested by Rutz it is necessary to recall the geodesic deviation equation. Let $\gamma$ be a Finsler geodesic and $\hat \gamma = (\gamma, \dot \gamma)$ be its lift to the tangent bundle. Tangent vectors of $\hat \gamma$ are horizontal, i.e., $\dot{\hat \gamma} = \dot{\gamma}^i\delta_i$. Moreover let $V$ be a deviation vector field on spacetime with canonical horizontal lift $\hat V=V^i\delta_i$. Then the geodesic deviation equation is
\begin{align}\label{eq:GeodDevi}
	(\nabla \nabla \hat V)|_{(\gamma, \dot \gamma)} = \textbf{R}(\dot{\hat \gamma},\hat V)\,.
\end{align}
The geodesic derivation operator $\textbf{R} = R^i{}_jdx^j\otimes\delta_i$ is derived from the curvature of the non-linear connection as
\begin{align}\label{eq:nlcurv}
	R^i{}_j = R^i{}_{jk}\dot x^k,\quad R^i{}_{jk}\dot{\partial}_i = [\delta_j, \delta_k]= (\delta_kG^i{}_j - \delta_jG^i{}_k)\dot{\partial}_i\,.
\end{align}
The non-homogenized Finsler Ricci scalar $R$ is given by its trace
\begin{align}\label{eq:FinslerRicciS}
	R = R^i{}_i = R^i{}_{ik}\dot x^k\,.
\end{align}
It is important to observe that the curvature tensors appearing here are defined solely in terms of the canonical Cartan non-linear connection. The Finsler linear connections, which one may define, are not entering here.

In case the Finsler-Lagrange function takes the form $L=g_{ij}(x)\dot{x}^i\dot{x}^j$, where $g_{ij}(x)$ are the components of a Lorentzian metric, the geometry of a Finsler spacetime $(M,L)$ becomes essentially the geometry of the pseudo Riemannian spacetime manifold $(M,g)$. The $L$-metric becomes the Lorentzian metric, the Cartan tensor vanishes, the non-linear connection coefficients and the non-linear curvature tensor become the Christoffel symbols and the Riemann curvature tensor of the Levi-Civita connection of $g$, up to a contraction with a velocity $\dot x$. Observe that the Finsler Ricci scalar becomes $R(x,\dot x) = -r_{jk}(x)\dot x^j \dot x^k$ and is not equal to the Riemannian Ricci scalar $r = r_{ij}g^{ij}$ in this case.

When we construct an action for Finsler gravity in Section \ref{sec:vari} we will work on the positive projective tangent bundle with $0$-homogeneous objects. On $\mathcal{A}_0$ we can introduce the $0$-homogenized Ricci scalar\footnote{$R_0$ is commonly denoted by $Ric$ in the literature. We choose the subscript $0$ here to indicate it is an object on the set $\mathcal{A}_0$.}
\begin{align}
	R_0 = \frac{1}{L}R\,,
\end{align}
which will be the key ingredient to the Lagrangian density defining the gravity action. Additionally we need a canonical invariant $0$-homogeneous volume form on $\mathcal{A}_0$, which is given by
\begin{align}\label{eq:volA0}
	\textrm{Vol}_{0} = \frac{1}{L^2} |\det(g^L)| \textrm{Vol} = \frac{1}{L^2} |\det(g^L)| dx^0\wedge \ldots \wedge dx^3 \wedge d\dot x^0 \wedge \ldots \wedge d\dot x^3 \,,
\end{align}
where we use the abbreviation $\textrm{Vol} = dx^0\wedge \ldots \wedge dx^3 \wedge d\dot x^0 \wedge \ldots \wedge d\dot x^3$ for the local Euclidean volume form. This volume form is indeed $0$-homogeneous w.r.t. $\dot x$, which can be seen from the fact that
\begin{align}
	\mathcal{L}_{\mathbb{C}}\textrm{Vol}_{0} = 0\,,
\end{align}
where the Liouville vector field $\mathbb{C} = \dot x^i \dot{\partial}_i$ is the generator of the homotheties $(x,\dot x)\mapsto (x,\lambda \dot x)$. During the derivation of the Finsler gravity field equation the following divergence formulas for horizontal and vertical vector fields, $X=X^i\delta_i$ and $Y=Y^i\dot{\partial}_i$, on $\mathcal{A}_0$, turn out to be very useful:
\begin{align}
\mathrm{div}(X) \textrm{Vol}_{0} = \mathcal{L}_X(\textrm{Vol}_{0}) \Leftrightarrow \mathrm{div}(X) &= (X^i{}_{|i} - P_{i}X^{i}),\text{ with } P_i = P^j{}_{ji}\,,\label{divergence_horizontal}\\
\mathrm{div}(Y) \textrm{Vol}_{0} = \mathcal{L}_Y(\textrm{Vol}_{0}) \Leftrightarrow \mathrm{div}(Y) &= (Y^i{}_{\cdot i} + 2 C_{i}Y^{i} - \frac{4}{L}Y^i\dot x_i),\text{ with } C_i = C^j{}_{ji}\,,\label{divergence_vertical}
\end{align}
which imply for any $f:\mathcal{A}_0 \rightarrow \mathbb{R}$,
\begin{align}\label{eq:difnabla}
\mathrm{div}(f \dot x^i\delta_i) = \nabla f\,.
\end{align}
These equations were obtained by direct calculation from the volume form $\textrm{Vol}_{0}$.

\section{Finsler spacetime dynamics}\label{sec:FSTdyn}
The demand that Finsler spacetime dynamics shall use the Finsler function as fundamental variable selects among the conjectured Finsler spacetime dynamics to the ones suggested in \cite{Rutz,Chen-Shen,Pfeifer}. The first field equation which took the Finsler function as fundamental variable and was itself a scalar equation on the tangent bundle on the manifold was obtained by Rutz \cite{Rutz}. It was argued that from the geodesic deviation equation, one finds the relevant curvature structure of spacetime which causes tidal forces between neighboring trajectories, and that its trace is a suitable approach as gravitational vacuum field equation. The same argument was applied in the pseudo-Riemannian case by Pirani to obtain the Einstein vacuum field equations \cite{Pirani}. Rutz's equation simply states that the canonical non-linear Finsler curvature scalar \eqref{eq:FinslerRicciS} vanishes
\begin{align}\label{eq:rutz}
R = 0\,.
\end{align}
It measures the trace of the geodesic deviation operator \eqref{eq:GeodDevi}, understood as a function of the Finsler-Lagrange function $L$ and its derivatives.

Action based Finsler field equations using the Finsler function as fundamental variable have been obtained by calculus of variations in \cite{Pfeifer} and \cite{Chen-Shen} independently, in the first case Finsler spacetimes, in the later case for positive definite Finsler spaces. The action employed is
\begin{align}\label{eq:fgravaction}
	S[L] = \int_{\Sigma\subset TM} \mathrm{vol}(\Sigma) R_{|\Sigma}\,,
\end{align}
where $\Sigma = \{(x,\dot x)\in TM|F(x,\dot x) = 1\}$ is the unit tangent bundle and $\mathrm{vol}(\Sigma)$ the volume form on $\Sigma$ defined from the Finsler metric. Variation with respect to $L$ yields the Finsler spacetime vacuum dynamics\footnote{Observe that in \cite{Pfeifer}, the Landsberg tensor, called $S^i{}_{jk}$ there, is defined with a different sign, $S^i{}_{jk} = - P^i{}_{jk}$.}
\begin{align}\label{eq:fgraveq}
	2R - \frac{L}{3}g^{Lij}R_{\cdot i \cdot j} + \frac{2L}{3}g^{Lij}( (\nabla P_{i})_{\cdot j} + P_{i|j} - P_{i}P_{j})&= 0\,.
\end{align}

In the particular case of a pseudo-Riemannian Finsler-Lagrange function determined by a Lorentzian metric $g$, Rutz's equation as well as the action based Finsler spacetime dynamics are equivalent to the Einstein vacuum equations $r_{ab} = 0$. For the action based spacetime dynamics it is possible to add a matter field action so that the resulting gravitational dynamics reduce  to the Einstein equations~\eqref{full_einstein_eq} \cite{Pfeifer}.

We will see in section \ref{sec:RutzVT} in detail that Rutz's equation has the disadvantage that it can not be obtained as an Euler-Lagrange equation. However, applying the variational completion algorithm developed in \cite{Voicu-Krupka} to \eqref{eq:rutz} yields the field equations \eqref{eq:fgraveq}. The analogue statement holds for the field equations $r_{ij} = 0$ and $G_{ij} = 0$ in general relativity: the former have the disadvantage that they can not be obtained as Euler-Lagrange equations, while the latter are the result of the variational completion algorithm applied to the former. Thus only the latter can be completed consistently to non-vacuum dynamics.

\section{The stage for a Finsler gravity action}\label{sec:FGraviAction}
As we have seen in the previous sections, all geometric objects in Finsler geometry possess homogeneity properties with respect to their dependence in $\dot x$. This means that they are equivariant under the action of a Lie group, which makes it more appropriate to describe them on a bundle that takes this equivariance into account. In the previous approaches to action based Finsler gravity equations \cite{Pfeifer,Chen-Shen} this fact was taken care of by constructing an action on the unit tangent bundle $\Sigma = \{(x,\dot x)\in TM|F(x,\dot x) = 1\}$. However this construction has the drawback, that the domain of integration depends on the dynamical variable one is interested in, and, in the case of a Lorentzian signature of the Finsler metric, is non-compact. As consequence action integrals, formulated as integrals over all of $\Sigma$ are generically infinite.

To avoid these problems we construct the action integral for Finsler gravity on compact subsets of the projective tangent bundle $PTM^+$ in Section \ref{sec:vari}. The advantage is that $PTM^+$ can be defined without any further structure on $TM$, and so is in particular independent of the Finsler function. Here we introduce the positive projective tangent bundle and how one can understand the Finsler function as section of an associated vector bundle over $PTM^+$. This construction allows a mathematically rigorous formulation of the Finsler gravity action as well as a technically precise derivation of the Euler-Lagrange equations from the action.

\subsection{The positive, or oriented, projective tangent bundle}
The \emph{positive projective tangent bundle} $PTM^+$, also called oriented tangent bundle in the literature \cite{Krupka-Notes}, can be constructed from the slit tangent bundle $\overset{\circ }{TM}:=TM\backslash \{0\}$ by identifying a ray $\left\{ (x,\lambda \dot{x})~|~\lambda >0\right\}$ as a single point. In other words it is defined by the equivalence relation
\begin{align}
(x,\dot{x})\ \sim\ (x,w) \quad\Leftrightarrow\quad w=\lambda \dot{x} \textrm{ for some } \lambda >0
\end{align}
between elements $(x,\dot{x})$ and $(x,w)$ in $\overset{\circ }{TM}$. To be precise
\begin{align}
PTM^{+}:=\{[(x,\dot{x})]_{\sim }~|~(x,\dot{x})\in \overset{\circ }{TM\}}\,.
\end{align}
For a $4$-dimensional base manifold $M$, the positive projective tangent bundle is itself a $7$-dimensional manifold, with manifold structure given by an atlas $\left\{ \left( U_{i}^{+},\varphi _{i}^{+}\right) ,\left( U_{i}^{-},\varphi _{i}^{-}\right) \right\}$, where, e.g., $U_{i}^{+}=\left\{ [(x^{0},\ldots,x^{3},\dot{x}^{0},\ldots,\dot{x}^{i},\ldots,\dot{x}^{3})]~|~\dot{x}^{i}>0\right\}$, $U_{i}^{-}=\left\{ [(x^{0},\ldots,x^{3},\dot{x}^{0},\ldots,\dot{x}^{i},\ldots,\dot{x}^{3})]~|~\dot{x}^{i}<0\right\}$ and we will denote collectively $\varphi _{i}^{+}$ and $\varphi _{i}^{-}$ as $(x^{j},u^{\alpha})$, $\alpha =1,2,3$:
\begin{align}\label{eq:ptmcoord}
(x^{i},u^{\alpha })=\left(x^{0},\ldots,x^{3},\dfrac{\dot{x}^{0}}{\dot{x}^{i}},\ldots,
\dfrac{\dot{x}^{i-1}}{\dot{x}^{i}},\dfrac{\dot{x}^{i+1}}{\dot{x}^{i}},\ldots,
\dfrac{\dot{x}^{3}}{\dot{x}^{i}}\right)\,.
\end{align}%
Alternatively, one can locally describe $PTM^+$ in homogeneous coordinates $(x^i,\dot x^i)$ in which one can perform calculations basically as on $TM$. To do so one has to ensure that the objects one is dealing with on $TM$ can be identified with well defined objects on $PTM^+$, in particular they must be $0$-homogeneous with respect to~$\dot x$~\cite{Chern-Shen-Lam}.

The manifold $PTM^{+}$ is compact and orientable. This can be seen easily, e.g., as $PTM^{+}$ is diffeomorphic to the unit sphere bundle of an arbitrarily chosen (positive definite) Riemannian metric on $M$.

The positive projective tangent bundle is defined without any reference to additional geometric structure on $\overset{\circ }{TM}$. However, if the slit tangent bundle is equipped with a classical, smooth and positive definite Finsler function, $PTM^{+}$ (then also called projective sphere bundle \cite{Shen}) is diffeomorphic to the unit tangent (or unit sphere) bundle $\Sigma$. For Finsler spacetimes, and in general for Finsler functions with associated metric of indefinite signature, such a global diffeomorphism does not exist. What however does exist is a diffeomorphism between certain compact subsets $D^+$ of $PTM^{+}$ and specific compact subsets $D$ of the unit tangent bundle $\Sigma \subset \overset{\circ }{TM}$. We will now construct such diffeomorphisms and relate the integration over admissible and non lightlike domains $D^+\subset PTM^+$ to the integration over $D\subset \Sigma$.

The first step towards this goal is to observe that $\overset{\circ }{TM}$ with the action of the multiplicative group $\mathbb{R}^*_+$
\begin{align}
	\cdot: \overset{\circ }{TM} \times \mathbb{R}_+^*  \rightarrow \overset{\circ }{TM},\quad (x,\dot x)\cdot \lambda = (x,\lambda\dot x)\,,
\end{align}
is a principal bundle over $PTM^{+}$. This can be seen easily, as $\mathbb{R}_{+}^*$ acts freely and transitively on the fibers of $\overset{\circ }{TM}$ relative to the
projection:
\begin{align}
\pi^+: \overset{\circ}{TM}\ \rightarrow PTM^+,\quad (x,\dot{x}) \mapsto [ x,\dot{x}] .  \label{projection_PTM+}
\end{align}
The local fibers of $(\overset{\circ }{TM},\pi ^{+},PTM^{+}, \mathbb{R}^*_+)$ are diffeomorphic to the 1-dimensional Lie group $\mathbb{R}^*_+$, whose action is generated by
the Liouville vector field $\mathbb{C}=\dot{x}^{i} \dot{\partial}_{i}$. That is, $\mathbb{C}$ is a $\pi ^{+}$-vertical vector field on $\overset{\circ }{TM}$.

Second, the projection $\pi^+$ allows us to treat differential forms on $PTM^+$ as certain particular differential forms on $\overset{\circ }{TM}$ via pullback. Let $\rho^+$ be a differential form on $PTM^+$, then $\rho:=(\pi^+)^*\rho^+$ is a basic form with respect to $\pi^+$, i.e. it satisfies:
\begin{itemize}
	\item equivariance with respect to the action of the Lie group $(\mathbb{R}^*_+,\cdot)$, or, in other words, it is $0$-homogeneous in $\dot x$
	\begin{align}\label{Lie deriv rho}
		\mathcal{L}_{\mathbb{C}}\rho = 0\,,
	\end{align}
	\item horizontality with respect to the projection $\pi^+$, which means that contracted with $\mathbb{C}$ it satisfies
	\begin{align}\label{contraction_rho_C}
		\mathbf{i}_{\mathbb{C}}\rho =0.
	\end{align}
\end{itemize}
Exterior differentiation of forms $\rho^+\in\Omega(PTM^+)$ can be carried out identically to exterior differentiation of the corresponding form $\rho\in\Omega(\overset{\circ }{TM})$ as $d\circ (\pi^+)^*=( \pi ^+)^*\circ d$.

The third step is to realize that $\pi^+$ is a diffeomorphism between a compact admissible subset
\begin{align}
D \subset \{(x,\dot x)\in \mathcal{A}_0| F(x,\dot x) = 1\}
\end{align}
of the level hypersurface of the $1$-homogeneous smooth map $F$ on $\mathcal{A}_0$ and its image $\pi^+(D)\subset PTM^+$. This is easy to see. The mapping $\pi ^{+}:D\rightarrow \pi ^{+}(D)$, $\pi ^{+}(x,\dot{x})=[x,\dot{x}]$ is obviously bijective. Differentiability is also immediate. Since, on $D$, $F=1\not=0,$ the inverse map $\alpha :\pi^{+}(D)\rightarrow D,$ $\alpha ([x,\dot{x}])=(x,\tfrac{\dot{x}}{F(x,\dot{x})})$ of $\pi ^{+}$ is also differentiable.


Consequently we can state
\begin{lemma}\label{lemm:int}
	For any compact domain $D^{+}$ on $PTM^{+}$ such that $F( (\pi^+)^{-1}(D^+))\neq 0$ and any $7$-form $\rho^{+}$ on $PTM^{+}$:
	\begin{align}\label{rel_integrals_SM_PTM}
	\underset{D^+}{\int }\rho^+ = \underset{D}{\int }\rho\,,
	\end{align}
	where $\rho =(\pi^+)^*\rho^+$ is a $0$-homogeneous differential form on $\overset{\circ }{TM}$ and $D=(\pi^+)^{-1}(D^+)\cap \Sigma$.
\end{lemma}
Hence we related integrals on $PTM^+$ to integrals on $\Sigma$. From the point of view of local computations, pulling back forms to $D \subset \Sigma \subset \overset{\circ }{TM}$ or working on $D^+$ in homogeneous coordinates are the same.

Observe that vector fields $X^+$ on $D^+\subset PTM^+$ can be identified with zero homogeneous vector fields $X$ on $(\pi^+)^{-1}(D^+)$, i.e.,\ vector fields satisfying $[X,\mathbb{C}] = 0$, via the tangent map $d(\pi^+)^{-1}$.

\subsection{Finsler spacetime action integrals on $PTM^+$}
On a Finsler spacetime $(M, L)$ integrals on $PTM^+$ can be understood as follows.

Any $7$-form $\rho^+$ can be decomposed into a product of the canonical volume form $dV_0^+$ on $PTM^+$ and a function. The canonical volume form can be obtained via the canonical \emph{Hilbert form} as follows. Consider the set
\begin{align}
	\mathcal{A}_0^+ = \pi^+(\mathcal{A}_0)\subset PTM^+\,.
\end{align}
An equivalent characterisation of $\mathcal{A}_0^+$ is $\mathcal{A}_0^+ = \pi^+(\Sigma \cap \mathcal{A})$.

The functions $l_i = \dot{\partial_i}F$ are well defined on $\mathcal{A}_0$ and, by their $0$-homogeneity with respect to $\dot x$, also on $\mathcal{A}_0^+$. This implies that the Hilbert form
\begin{align}
	\omega = l_i dx^i\,,
\end{align}
is a well defined coordinate invariant $1$-form on $\mathcal{A}_0^+$ and
\begin{align}
	\omega \wedge d\omega \wedge d\omega \wedge d\omega \neq 0
\end{align}
is a well defined $7$-form on $\mathcal{A}^+_0$ \cite{Chern-Shen-Lam}. Hence, a coordinate invariant, well-defined volume form on $\mathcal{A}_0^+$ is given by
\begin{align}
	dV_0^+ = \frac{1}{3!}\omega \wedge d\omega \wedge d\omega \wedge d\omega\,.
\end{align}
In local homogeneous coordinates on $\mathcal{A}_0^+$ it can be expanded as
\begin{align}
	dV_0^+ = \frac{|\det g^L|}{L^2} \mathbf{i}_{\mathbb{C}}(dx^0\wedge \ldots \wedge dx^3 \wedge d\dot x^0 \wedge \ldots \wedge d\dot x^3)\,.
\end{align}
The pullback of $dV_0^+$ by $\pi^+$ yields a $7$-form on $\mathcal{A}_0$, which can be expressed in terms of the $0$-homogeneous volume form $\textrm{Vol}_0$, see \eqref{eq:volA0},
\begin{align}
	 dV_0 := (\pi^+)^*dV_0^+ = \mathbf{i}_{\mathbb{C}}\textrm{Vol}_0= \frac{|\det g^L|}{L^2}(\mathbf{i}_{\mathbb{C}}\textrm{Vol})\,.
\end{align}
Thus in local homogeneous coordinates the coordinate expressions of $dV_0$ and $dV_0^+$ are identical \cite{Chern-Shen-Lam} and by abuse of notation we do not display the pullback explicitly in each expression.

Note that for $0$-homogeneous vector fields $X$ on $\overset{\circ }{TM}$, their divergence with respect to the volume form $dV^+_0$ resp. $dV_0$ is given by the same expressions as the divergences with respect to the volume form $\textrm{Vol}_0$ displayed in \eqref{divergence_horizontal} and~\eqref{divergence_vertical}.

Integrals on compact domains $D^+\subset \mathcal{A}^+_0$ can be written as integrals on $D\subset(\Sigma \cap \mathcal{A})$
\begin{align}\label{eq:actioninteg}
\underset{D^{+}}{\int }\ f\ dV_0^+ =\underset{D}{\int }\ f\circ \pi^+\ dV_0\,,
\end{align}
where $f$ is a function on $PTM^+$. In local homogeneous coordinates the expressions of $f$ and $f\circ \pi^+$, which is a $0$-homogenous function on $TM$, are identical. For us $f$ will be the Lagrange function which we will obtain from variational completion of Rutz's equation in Section \ref{sec:RutzVT}.

The last technical construction to write down the Finsler gravity action in Section \ref{sec:vari}, is to understand our dynamical variable $L$ as a section of a fibered manifold \cite{Krupka-book}. It turns out that the most natural such choice is an associated bundle to the principal bundle $(\overset{\circ }{TM},\pi^+, PTM^+)$, which we have already discussed. By the definition of the following equivalence relation on $\overset{\circ }{TM} \times \mathbb{R}^*_+$
\begin{align}\label{equivalence_rel_Y}
	(x,\dot x, y)\sim (x,\lambda \dot x, \lambda^2 y)
\end{align}
for all \(\lambda > 0\) we can construct the associated bundle $(Y, \pi_Y, PTM^+)$, with
\begin{align}\label{eq:Y}
	Y := (\overset{\circ }{TM} \times \mathbb{R}^*_+)/_{\sim},\quad \pi_Y([x,\dot x,y]) = [x,\dot x]
\end{align}
and fiber $\mathbb{R}^*_+$. Homogeneous coordinates corresponding to a fibered chart on this manifold then are $(x^i, \dot x^i , y)$. It is now easy to see that there is a one-to-one correspondence between $2$-homogeneous maps $L:TM \rightarrow \mathbb{R}$ and sections $\gamma$ of $Y$
\begin{align}
	L\mapsto \gamma: PTM^+ \rightarrow Y, \gamma([x,\dot x]) = [x,\dot x, L(x,\dot x)]\,.
\end{align}
This can be checked as follows. The mapping is well-defined, as, for any $(x,\lambda \dot{x})\in [x,\dot{x}],$ we have $[x,\lambda \dot{x},L(x,\lambda \dot{x})]=[x,\dot{x},L(x,\dot{x})]$ by virtue of \eqref{equivalence_rel_Y}. Its injectivity and surjectivity are immediate.

An important Lemma, inspired by a similar statement on Finsler spaces found in \cite{Chen-Shen}, which allows us to evaluate and manipulate the action integral later is
\begin{lemma}\label{lem:intbdry}
	Let $(M,L)$ be a general Finsler spacetime and let $f$ be a $0$-homogeneous function on $\mathcal{A}_0$. Moreover let $X$ be the vertical vector field $X = (L g^{Lij} f_{\cdot i})\dot{\partial}_j$, then the following identities on $\mathcal{A}_0^+$ hold:
	\begin{align}
		[g^{Lij} (Lf)_{\cdot i\cdot j} - 8 f] dV^+_0 &= d (\mathbf{i}_X dV^+_0)\label{eq:L3e1}\,,\\
		[L^{-1}g^{Lij}(L^2 f)_{\cdot i\cdot j} - 24 f]dV^+_0 &= d(\mathbf{i}_X dV^+_0)\label{eq:L3e2}\\
	\intertext{and}
		(g^{Lij} - 4 L^{-1} \dot x^i \dot x^j)(Lf)_{\cdot i\cdot j} dV^+_0 &= d (\mathbf{i}_X dV^+_0)\,.\label{eq:gxx}
	\end{align}
\end{lemma}
The proof of the Lemma can be found in Appendix \ref{app:Lemmaint}.

An important consequence of this Lemma is that, for functions $\phi_{ij} = \phi_{ij}(x)$, integrating $4\phi_{ij}(x)\dot x^i \dot x^j L^{-1}$ is identical to integration of $\phi_{ij}g^{Lij}(x,\dot x)$ up to a boundary term. To see this simply consider functions $f = \frac{1}{2}\dot x^i \dot x^j \phi_{ij}(x) L^{-1}$ in~\eqref{eq:gxx} to equate
\begin{align}\label{eq:totdiv}
	(g^{Lij} - 4 L^{-1} \dot x^i \dot x^j)\phi_{ij} dV^+_0 &= d (\mathbf{i}_X dV^+_0)\,.
\end{align}
This consequence of Lemma  \ref{lem:intbdry} will be useful to analyze the matter coupling to gravity in the Language of Finsler geometry in the future.

This completes the discussion of the technical ingredients to apply the variational completion algorithm to Rutz's equation and to analyze the resulting field equations.

\section{Rutz's equation and its Vainberg-Tonti Lagrangian}\label{sec:RutzVT}
\emph{Canonical variational completion} \cite{Voicu-Krupka} is a powerful algorithm to assess whether a certain set of field equations can be obtained by variation of an action functional or not. In case it can locally be obtained by variation, the algorithm determines the action, and, in the contrary case, the algorithm determines a standard term to be added to the equations to make them variational. Before we apply the technique to Rutz's equation \eqref{eq:rutz}, we recall its main steps.

Consider a set of $m$ partial differential equations (PDEs) of order $r$ in the independent variables $x^A\in \mathbb{R}^N$ (regarded as coordinates in a local chart $U$ on some manifold X)  and the dependent variables $y^{\mu}=y^{\mu}(x^A)$
\begin{align}\label{general_PDE_system}
	\varepsilon_{\sigma} (x^A,y^{\mu},y^{\mu}{}_{A_1},\ldots,y^{\mu}{}_{A_1 \cdots A_r}) = 0\,,
\end{align}
where $A=1,\ldots,N$ and $\mu,\sigma=1,\ldots,m$. The subscripts on $y^{\mu}$ denote partial differentiation, i.e.\ $y^{\mu}_A = \partial_{x^A}y^{\mu}$ and so on. Note that the number $m$ of equations coincides with the number of dependent variables.

From equations \eqref{general_PDE_system}, we can build, on a given coordinate chart, the so-called \emph{Vainberg-Tonti Lagrangian density}
\begin{align}\label{VT_Lagrangian}
	\mathcal{L}= y^{\sigma} \int_0^1 \varepsilon_{\sigma}(x^A,ty^{\mu},ty_{A_1}^{\mu},\ldots,ty_{A_{1} \cdots A_{r}}^{\mu})dt\,.
\end{align}
The Vainberg-Tonti Lagrangian density $\mathcal{L}$ is the "closest" Lagrangian density to our PDE system, in the sense that, if equations \eqref{general_PDE_system} are locally variational, i.e., if they can be locally
written as the Euler-Lagrange equations attached to some Lagrangian density, then, this Lagrangian density is, up to a total derivative term, $\mathcal{L}$.

The quantities which measure the departure of the original PDE system of interest from being variational are the components of the so called \textit{Helmholtz form}
\begin{align}\label{eq:hhol}
	H_{\sigma }:=E_{\sigma }-\varepsilon _{\sigma }\,,
\end{align}
where
\begin{align}
	E_{\sigma } =\frac{\partial \mathcal{L}}{\partial y^{\sigma }}-d_{A_1}\frac{\partial \mathcal{L}}{\partial y^{\sigma}{}_{A_1}}+ \ldots +(-1)^{r} d_{A_1} \ldots d_{A_r}\dfrac{\partial \mathcal{L}}{\partial y^{\sigma}{}_{A_{1} \cdots A_{r}}}
\end{align}
are the Euler-Lagrange expressions of $\mathcal{L}$ formulated in terms of total derivative operators $d_A$ with respect to $x^A$. The following result is the key to examine if the original PDEs we started with were variational, see \cite{Voicu-Krupka}:
\begin{proposition}
	The PDE system \eqref{general_PDE_system} is locally variational if and only if, in any local chart the Helmholtz conditions
	\begin{align}
	H_{\sigma }=0,\quad \sigma =1, \ldots ,m\,,  \label{Helmholtz_conds}
	\end{align}
	hold.
\end{proposition}
The \emph{canonical variational completion}, see again \cite{Voicu-Krupka}, of the PDE system $\varepsilon_{\sigma} = 0$ are the PDEs
\begin{align}\label{canonical var completion}
	E_{\sigma} = 0\,.
\end{align}
The term \textit{canonical} comes from the fact that, adding to the left hand sides of equations \eqref{canonical var completion} any locally variational term will still result in a variational PDE system. But \eqref{canonical var completion} are the closest variational equations to the initial ones, as indicated by \eqref{eq:hhol} and \eqref{Helmholtz_conds}. In particular, they are completely determined by the functions $\varepsilon_{\sigma}$ alone.

To illustrate the framework notice that a typical example of variational completion is the derivation of the completion of the Einstein vacuum equations. On a Lorentzian manifold $(M,g),$ the canonical variational
completion of the equations $r_{ij}=0$ are the full equations $r_{ij}-\frac{1}{2}rg_{ij}=0$ \cite{Voicu-Krupka}.

Finally let us apply the canonical variational completion to Rutz's equation \eqref{eq:rutz}. The setup is
\begin{align}
	X=D^+ \subset PTM^+,\quad x^A=(x^i,\dot x^i),\quad y = y^{1} = L\ \textrm{(i.e.\ $\mu=1$)}\,,
\end{align}
where the coordinates $x^A$ are again homogeneous coordinates \cite{Chern-Shen-Lam}.

In order to get a correct scalar density, let us multiply Rutz's equation by  $\left\vert \det g^L\right\vert$. In addition we multiply it by $L$ to an arbitrary power $\alpha$, in order to be able to adjust the homogeneity of the desired Lagrange density later. Thus, Rutz's equation becomes:
\begin{align}
\varepsilon :=L^{\alpha}R\left\vert \det g^L\right\vert =0\,.
\label{Rutz_eqn_processed}
\end{align}%
Taking into account the local expression of $R=R^i{}_{ij}\dot x^j$, we see that $\varepsilon=\varepsilon(\dot x^i, L,L_{,i},L_{\cdot i},\ldots,L_{\cdot i \cdot j \cdot k \cdot l})$ depends on $L$ and its partial derivatives up to order four.

For each coordinate neighborhood on $D^+$, we find the local Vainberg-Tonti Lagrangian as:
\begin{align}\label{VT_Lagrangian_rough}
\mathcal{L}= L \int_0^1 \varepsilon (\dot x^i, tL,tL_{,i},tL_{\cdot i},\ldots,tL_{\cdot i\cdot j \cdot k \cdot l})dt\,.
\end{align}
To evaluate the integral in the Lagrangian we note that, with respect to the fiber homotheties $L\mapsto \tilde{L}:=tL,$ the metric tensor components $g^L_{ij}$ and the inverse metric transform as
\begin{align}
	\tilde g^{\tilde L}_{ij} = t g^L_{ij},\quad \tilde g^{\tilde L ij} = t^{-1} g^{L ij}\,.
\end{align}
The geodesic spray \eqref{eq:geodspray} behaves thus as
\begin{align}
	\tilde G^i(\tilde L,\tilde L_{,i},\tilde L_{, i \cdot j},\tilde L_{\cdot i \cdot j}) = G^i(L,L_{,i},L_{, i \cdot j},L_{\cdot i \cdot j})\,,
\end{align}
which implies the same behavior for the curvature and most importantly for the Finsler Ricci scalar
\begin{align}
	\tilde R^i{}_{jk} = R^i{}_{jk},\quad \tilde R = R\,.
\end{align}
The last missing ingredient in the Vainberg-Tonti Lagrangian is the volume form factor which, by the fact that we are considering a four dimensional manifold $M$, transforms as
\begin{align}
	|\det \tilde g^{\tilde L}| = t^4 |\det g^L|\,.
\end{align}
Employing the scaling behaviors just discussed we find the desired Lagrangian density
\begin{align}
\mathcal{L}= L^{\alpha+1} R |\det g^L| \int_0^1 t^{\alpha+4} dt = \frac{1}{\alpha+5}L^{\alpha+1} R |\det g^L|\,.
\end{align}
In order to correctly define a Lagrangian on $PTM^+$, we must construct a $4$-form of the type $\rho = f dV^+_0 $, with a zero homogeneous $f$, as discussed in \eqref{eq:actioninteg}. To achieve this the above expression for $\mathcal{L}$ must be $(-4)$-homogeneous and so $\alpha$ must be chosen to be $-4$, since $R$ is $2$-homogeneous and $\det g^L$ is $0$-homogeneous.

Thus we conclude that the Lagrange density which yields the variationally completed field equations to Rutz's equation is
\begin{align}\label{eq:VTLRutz}
	\mathcal{L} = L^{-3} R |\det g^L| = R_0 L^{-2}|\det g^L|\,.
\end{align}
This Lagrange density coincides with the ones suggested in \cite{Pfeifer} and \cite{Chen-Shen} (for positive definite Finsler spaces), here derived by the means of variational completion.

Following the canonical variational completion algorithm we found that if Rutz's equation \eqref{eq:rutz} is variational, then the Lagrangian from which it shall be obtained by variation is given by \eqref{eq:VTLRutz}. What we will find in the next section is that the Euler-Lagrange equation of \eqref{eq:VTLRutz} does not coincide with Rutz's equation, so Rutz's equation can not be variational but must be variationally completed by the terms we will find next.

\section{Finsler gravity action and its first variation}\label{sec:vari}
The last step in the variational completion algorithm is to check whether the seed equation \eqref{general_PDE_system} can be obtained by variational calculus from the action defined by its Vainberg-Tonti Lagrangian \eqref{VT_Lagrangian_rough}. If so, the seed equation itself is variational, if not we find the closest variational completion of the seed equation.

The classical variational principle, \cite{Krupka-book}, requires the existence of a fibered manifold $\left( Y,\pi ,X\right) ,$ $\dim X=n,$ $\dim Y=m+n.$ The manifold $Y$ is called the configuration manifold and $X,$ the base (typically - but not necessarily - spacetime) manifold. Sections $\gamma \in \Gamma (Y)$ will be interpreted as \textit{fields}. Deformations of a field (section) are given by 1-parameter groups of fibered automorphisms, generated by projectable vector fields on $Y$.

In this setting, a \textit{Lagrangian} of order $r$ is regarded as a horizontal differential form on the jet bundle $J^{r}Y.$ Denoting by $(x^{A},y^{\sigma },y_{~A}^{\sigma },\ldots,y_{~A_{1} \cdots A_{n}}^{\sigma })$ the fibered coordinates on $J^{r}Y,$ a Lagrangian is locally expressed as $\Lambda =\mathcal{L}d^{n}x^{A},$ where the Lagrangian density $\mathcal{L}=\mathcal{L}(x^{A},y^{\sigma },y_{~A}^{\sigma},\ldots,y_{~A_{1} \cdots A_{n}}^{\sigma })$ is a real-valued function on some open subset of $J^{r}Y.$

\subsection{The Finsler gravity action on $PTM^+$}
Consider an arbitrary connected compact subset $D^+ \subset PTM^+$. Without loss of generality, we are looking for functions L that are smooth and positive on $(\pi^+)^{-1}(D^+)$. Moreover consider the Vainberg-Tonti Lagrangian \eqref{eq:VTLRutz} we constructed in the previous section. Then, the action associated to the Lagrangian $\mathcal{L}$ and to the compact domain $D^+$ is the mapping $S_{D^+}:\Gamma (Y)\rightarrow \mathbb{R},$ $\gamma \mapsto S_{D^+}(\gamma )$ (recall the definition of $Y$ in \eqref{eq:Y}), given by
\begin{align}
	S_{D^+}(\gamma )=\underset{D^+}{\int } J^{4}\gamma^*\Lambda \,,
\end{align}
where \(J^4\gamma^*\Lambda\) is the pullback of \(\Lambda\) to $PTM^+$, along the jet prolongation $J^{4}\gamma$ of \(\gamma\) to the jet bundle \(J^4Y\).

In local homogeneous coordinates on $PTM^+$, $\gamma: [x,\dot x] \mapsto [x,\dot x, L(x,\dot x)]$ can be expressed as $(x,\dot x) \mapsto (x,\dot x, L(x,\dot x))$ and the action becomes
\begin{align}\label{eq:fgravact}
	S_{D^+}(\gamma ) = \int_{D^+} R_0 dV_0^+ = \int_{D^+} \frac{R}{L^3} |\det g^L|\  \mathbf{i}_{\mathbb{C}}(dx^0\wedge \ldots \wedge dx^3 \wedge d\dot x^0 \wedge \ldots \wedge d\dot x^3)\,.
\end{align}
Equivalently, according to Lemma \ref{lemm:int}, this integral can be formulated on $D\subset (\Sigma \cap \mathcal{A})$.

\subsection{Derivation of the field equations}
Take an arbitrary vertical vector field on $Y$, with support strictly contained in $D^+$ and denote by $\{ \Phi _{t}\} $ its 1-parameter group of fibered automorphisms of $Y$. The deformed sections:
\begin{align}
\bar{\gamma}:=\Phi _{t}(\gamma): [x,\dot x] \rightarrow [x,\dot x,\bar L(x,\dot x)]
\end{align}
automatically correspond to equivariant (i.e. $2$-homogeneous) functions $\bar{L}$ on non-lightlike, admissible domains of $TM$. To perform the variation we define
\begin{equation}\label{rel_L_bar}
\delta L=\dfrac{d\bar{L}}{dt}|_{t=0}:=2v\,.
\end{equation}
This also implies the fact that the functions $v$ have to be $2$-homogeneous in $\dot{x}$, as well as that $v$ and its partial derivatives will vanish on the boundary $\partial D^+$.
Moreover, for small enough $t$, the signature of the corresponding Hessian remains the same, i.e., $\bar{L}$ is a spacetime Finsler function.

The first variation of the action $S_{D^+}(\gamma )$ is:
\begin{align}\label{extremal_def}
\delta S_{D^+}(\gamma ):=\dfrac{d}{dt}|_{t=0}(S_{D^+}(\bar{\gamma}))\,.
\end{align}
Critical points, or \textit{extremals, \cite{Krupka-book},} are defined by the condition that, for \textit{any} admissible, non-lightlike compact domain $D^+\subset PTM^{+}$ and \textit{any} $v$ with support contained in $D^+$, $\delta S_{D^+}(\gamma )=0$.

The expression $\delta S_{D^+}(\gamma)$ will be split into three integrals
\begin{equation}
\delta S_{D+}(\gamma )=(I_{1}+I_{2}+I_{3}) \label{sum_I}\,,
\end{equation}
which one obtains from differentiation in turns of the last expression of \eqref{eq:fgravact}. The integrals are
\begin{align}
	I_1 
	&= \int_{D^+} L^2 \left.\frac{d}{dt}\frac{1}{\bar L^3}\right|_{t=0}\ R\ dV_0^+\,,\\
	I_2 
	&= \int_{D^+} \frac{1}{L} \left.\frac{d\bar R}{dt}\right|_{t=0}\ dV_0^+\,,\\
	I_3 
	&=\int_{D^+} \frac{R}{L} \frac{1}{|\det g^L|}\left.\frac{d |\det \bar g^{\bar L}|}{dt}\right|_{t=0}\ dV_0^+\,.
\end{align}

The first integral is easily evaluated to be
\begin{align}
	I_1 = - \int_{D^+} 6 \frac{R}{L} \frac{v}{L}\ dV_0^+\,.
\end{align}
The other two integrals involve rather lengthy calculation which we will display in detail in Appendix \ref{app:I2I3} and yield
\begin{align}
	I_2 &= - \int_{D^+} 2g^{Lij}(P_{i|j} - P_iP_j + (\nabla P_i)_{\cdot j})\frac{v}{L}\ dV_0^+\,, \label{eq:int2}\\
	I_3 &= \int_{D^+} g^{Lij}R_{\cdot i\cdot j}\frac{v}{L}\ dV_0^+\,\,.\label{eq:int3}
\end{align}
Thus  finally we obtain that the extremal points of the Finsler gravity vacuum action \eqref{eq:fgravact}, formulated on the bundle $Y$ with positive projective tangent bundle as base manifold, must satisfy
\begin{align}\label{eq:fgraveq2}
	\delta S_{D+}(\gamma ) = \int_{D^+} \bigg[\frac{1}{2}g^{Lij}R_{\cdot i\cdot j} - 3 \frac{R}{L} - g^{Lij}(P_{i|j} - P_iP_j + (\nabla P_i)_{\cdot j})\bigg]\frac{2v}{L}dV_0^+ = 0\,,
\end{align}
which leads us to formulate
\begin{theorem}
	Critical points of the Finsler gravity action \eqref{eq:fgravact} formulated on subsets of the positive projective tangent bundle $PTM^+$ are given by the equation:
	\begin{align}\label{eq:fgrav}
		\frac{1}{2}g^{Lij}R_{\cdot i\cdot j} - \frac{3}{L} R - g^{Lij}(P_{i|j} - P_iP_j + (\nabla P_i)_{\cdot j}) = 0\,.
	\end{align}
\end{theorem}
Once a solution $L$ of this equation is found, it holds on the set of admissible non-lightlike vectors of $L$, which we denoted by $\mathcal{A}_0^+$.
This equation is identical to the one found in \cite{Pfeifer} on the unit tangent bundle and in \cite{Chen-Shen} for positive definite Finsler spaces. The important new ingredients here are that the integration domains on $PTM^+$ are compact and do not depend on the Finsler-Lagrange function as well as that the Lagrange density used in the action \eqref{eq:fgravact} was obtained by variational completion in Section \ref{sec:RutzVT}.

\section{The pseudo-Riemannian case}\label{sec:metric}
Before concluding, we exemplify our findings for pseudo-Riemannian Finsler geometries $L(x,\dot x)= g_{ij}(x)\dot x^i \dot x^j$, which are defined by a metric $g = g_{ij}(x)dx^i\otimes dx^j$ with Lorentzian signature.

For such Finsler-Lagrange functions the Landsberg tensor vanishes $P_i = 0$, the components of the Finsler metric become identical to the components of the Lorentzian metric $g^L_{ij}(x,\dot x) = g_{ij}(x)$, the Finsler Ricci scalar is the contracted Ricci tensor $R = - r_{ij}(x)\dot x^i \dot x^j$ and thus, multiplying by $L$, the Finsler gravity equation becomes
\begin{align}
	0 = 3 R - \frac{L}{2}g^{Lij}R_{\cdot i\cdot j} = - 3 r_{ij}(x)\dot x^i \dot x^j + (g_{ij}(x)\dot x^i \dot x^j)r(x)\,,
\end{align}
where $r(x) = r_{ij}(x)g^{ij}(x)$ is the Riemannian Ricci scalar. Taking a second derivative with respect to $\dot x$ the equation leads to
\begin{align}\label{eq:prvac}
	3 r_{ij} - g_{ij}r = 0\,,
\end{align}
which implies $r=0$. Hence we find equivalence to the Einstein vacuum equations $r_{ij}=0$.

After having found that the field equations have the correct pseudo-Riemannian limit we now demonstrate two further properties, which ensure consistency of the Finsler gravity action in the pseudo-Riemannian limit with the Einstein--Hilbert action of general relativity. First we will provide a clear relation between them, second, we show consistency with the Bianchi identities. Hence, when matter is coupled, no contradiction with the usual Bianchi identities will appear, as one might worry when looking at the factor of $3$ appearing where one might have expected a $2$ in \eqref{eq:prvac}.

\begin{proposition}\label{prop:dyneq}
	For the class of quadratic Finsler spacetime Lagrange functions, $L= g_{ij}(x)\dot x^i \dot x^j$ the Finsler gravity Lagrangian $\Lambda = R_0(x,\dot x) dV^+_0$ is identical to the Einstein--Hilbert Lagrangian of general relativity up to a boundary term and a volume factor.
\end{proposition}
\begin{proof}
	For quadratic Finsler spacetime Lagrange functions, $L= g_{ij}(x)\dot x^i \dot x^j$ the Finsler gravity Lagrangian is
	\begin{align}
		\Lambda = R_0(x,\dot x) dV^+_0 = - r_{ij}(x)\frac{\dot x^i \dot x^j}{L} dV^+_0\,.
	\end{align}
	Applying Lemma \ref{lem:intbdry}
	, with $f = - r_{ij}(x)\dot x^i \dot x^j L^{-1}$ and $X=L g^{ij} f_{\cdot j}\dot{\partial}_i$ yields
	\begin{align}\label{eq:FGtoEH}
	\Lambda = R_0(x,\dot x) dV^+_0 = - \frac{1}{4} r(x) dV^+_0 - \frac{1}{8} d(\mathbf{i}_X dV^+_0)\,.
	\end{align}
	When integrating over $D^+$ and performing the fiber integration in the first term on the r.h.s., it becomes the Einstein-Hilbert action multiplied by a finite volume term, while the second term is a pure boundary term, hence irrelevant for the derivation of the Euler-Lagrange equations. This derivation displays the relation between the Finsler gravity action and the Einstein-Hilbert action of general relativity.
\end{proof}

\begin{proposition}\label{prop:BI}
	Let $g$ be a Lorentzian metric and $L= g_{ij}(x)\dot x^i \dot x^j$. On a pseudo-Riemannian Finsler spacetime $(M,L)$, diffeomorphism invariance of the Finsler gravity action yields the contracted Bianchi identities up to a boundary term.
\end{proposition}
\begin{proof}
	Infinitesimal diffeomorphisms on $M$ are generated by vector fields $\xi = \xi^i(x)\partial_i$. Their action on the Finsler-Lagrange function $L$ is generated by a trivial lift of $\xi^C = \xi^i(x)\partial_i + \dot x^j \xi^i_{,j}\dot{\partial}_i$ to the configuration manifold $Y$. That is the quantity $v$ in \eqref{eq:fgraveq2} is given by $2v = \xi^C(L) = 2 \xi_{i|j}\dot x^i \dot x^j$, and the Euler Lagrange form $E(L)$ is
	\begin{align}
	E(L) = - \left(r(x) - 3 \frac{\dot x^m \dot x^n r_{mn}(x)}{L}\right)\frac{4 \xi_{i|j}\dot x^i \dot x^j}{L} dV_0^+\,.
	\end{align}
	Here the appearing Chern-Rund covariant derivative is identical to the Levi-Civita covariant derivative of the metric defining the Finsler-Lagrange function. We now prove that $E(L)$ is a total derivative expression up to the covariant divergence of the Einstein tensor, i.e. up to the contracted Bianchi identities.

	Applying equation \eqref{eq:totdiv} to the first term of $E(L)$ yields
	\begin{align}\label{eq:EL1}
		-4 r L^{-1} \xi_{i|j}\dot x^i \dot x^j\ dV_0^+ = - r \xi^i_{|i} \ dV_0^+ + d\Xi\,,
	\end{align}
	where $\Xi$ stands for total derivative terms appearing in Lemma \ref{lem:intbdry}. Next apply first equation \eqref{eq:L3e2} and then \eqref{eq:gxx} to the second term to obtain
	\begin{align}\label{eq:EL2}
		12 L^{-2} \dot x^m \dot x^n r_{mn} \xi_{i|j}\dot x^i \dot x^j\ dV_0^+
		&= \left( r \xi_{i|j} + 2 r_{im}(\xi^m_{|j} + \xi^l_{|n}g^{nm}g_{lj}) + r_{ij}\xi^m_{|m} \right)L^{-1}\dot x^i \dot x^j \ dV_0^+  + d\Upsilon\nonumber\\
		&= \left(\frac{1}{2}r\xi^m_{|m} + r^{ij}\xi_{i|j}\right)\ dV_0^+ +  d\zeta\,,
	\end{align}
	where $d\Upsilon$ and $d\zeta$ stand for total derivative terms appearing in Lemma \ref{lem:intbdry}. Adding \eqref{eq:EL1} and \eqref{eq:EL2} finally combines to
	\begin{align}
		E(L) = \left(-\frac{1}{2}g^{ij}r + r^{ij} \right) \xi_{i|j} dV_0^+ + d(\Xi + \zeta)\,,
	\end{align}
	which finally proves that $E(L)$ is a total derivative up to the contracted Bianchi identities.
\end{proof}

In an upcoming article we will investigate matter coupling and diffeomorphism invariance of the general Finsler gravity action in detail.

\section{Conclusion}\label{sec:conc}
To enrich and focus the discussion about a proper viable Finsler generalization of the Einstein equations we presented strong arguments which identify the equation \eqref{eq:fgrav} as the simplest, mathematically consistent, action based gravitational field equation, determining the Finsler function of a Finsler spacetime.

Starting from the physical argument that gravity causes the tidal forces in geodesic deviation, Rutz identified the Cartan non-linear curvature tensor as the relevant mathematical object which encodes the gravitational interaction on the basis of Finsler geometry. Here we proved, by using canonical variational completion, that the field equation suggested by her in \cite{Rutz} can not be obtained by the means of calculus of variations, and that the field equations suggested in \cite{Pfeifer} are the variational completion of the former ones. In order to rigorously apply this formalism, we formulate the Finsler gravity action on a configuration space sitting over the positive projective tangent bundle of the spacetime manifold. The positive projective tangent bundle is compact and does not depend on the Finsler function to be determined. Further, we showed that for pseudo-Riemannian metrics, the gravity action conjectured in \cite{Pfeifer} and found by variational completion here becomes the Einstein--Hilbert action up to a boundary term. This result may be understood as an additional argument supporting the action we constructed. Accordingly, consistency of equation \eqref{eq:fgrav} with Bianchi identities is established as a result of the invariance of the action to diffeomorphisms of the spacetime manifold.

Our results open the possibility for different further directions of research. One natural question is what generalizes the contracted Bianchi identities from the pseudo-Riemannian setting in general Finsler spacetime. This question is closely connected to the conservation equation for matter actions and the construction of a consistent matter coupling. The latter may be expressed either in terms of certain tensor or spinor fields, or as a kinetic fluid. Finally, to understand the predictivity of a Finsler gravity theory it is necessary to understand the initial value problem of the Finsler gravity equation. The correct initial value formulation for field equation on the projective tangent bundle must be constructed and the Finsler gravity equation shall be cast into an initial value form.

Regarding the matter coupling, one possibility is reconsider the tangent bundle matter actions suggested in \cite{Pfeifer} and reformulate them on the projective tangent bundle, in the same fashion as we did with the gravitational action. A new approach to matter couplings, which shall be investigated in the future, is offered by Lemma \ref{lem:intbdry}. It allows us to rewrite contractions with a metric into an integral over contractions with velocities and hence, to rewrite kinetic terms in the usual matter field actions on spacetime into matter field actions on the projective tangent bundle in a canonical way. The most promising outlook for a coupling between matter and a Finslerian tangent bundle geometry is the direct coupling of a kinetic gas to the geometry of spacetime. A kinetic gas is directly described on the tangent bundle in terms of one-particle distribution functions \cite{Ehlers2011}, which can naturally be described in the Finsler language~\cite{Hohmann:2015ywa}. In its standard formulation the gas back reacts to gravity via averaging, since its energy momentum-tensor, which couples to gravity via the Einstein equations, is obtained by averaging over the velocities of the constituents of the gas.  The reformulation of the dynamics of the kinetic gas on the positive projective tangent bundle allows us to directly couple it to the Finsler gravity and omit the averaging procedure. This offers the perspective of a more precise description of gravitating kinetic gases and its applications to cosmology for possible insight to dark energy as averaging effect.

\begin{acknowledgments}
The authors are grateful to Bin Chen, Miguel Sánchez Caja, Miguel Angel Javaloyes, Demeter Krupka and Volker Perlick for useful discussion about the matter of this article. Moreover the authors like to thank the anonymous referee for the very good and detailed comments and suggestions, which helped to improve the article. MH and CP were supported by the Estonian Ministry for Education and Science through the Institutional Research Support Project IUT02-27 and Personal Research Funding Grants PUT790 and PRG356, as well as the European Regional Development Fund through the Center of Excellence TK133 ``The Dark Side of the Universe''. NV was supported by a local grant of the Transilvania University. This article is based upon work from COST Action CANTATA (CA 15117), supported by COST (European Cooperation in Science and Technology).
\end{acknowledgments}

\appendix

\section{Convexity of $\mathcal{T}_x$}\label{app:Conv}
We claimed that our definition of Finsler spacetimes, Definition \ref{def:FST}, ensures the existence of a convex set $\mathcal{T}_x \subset T_xM$ for all $x\in M$. Here we give a proof of this statement.

First we note that $\partial \mathcal{T\subset N}$. Since $TM$ is a manifold, it behaves locally like $\mathbb{R}^{2n}$. In particular it is locally connected. As the function $L$ is assumed to be everywhere continuous and, on $\mathcal{T}$, we have by definition $L>0$, it follows that on the boundary $\partial\mathcal{T}$ we can only have $L\geq 0$. Assume that there exists a point $(x,\dot x)\in \partial \mathcal{T}$ such that $L(x,\dot x)>0$. It is sufficient to notice that there exists an open connected neighbourhood $B$ of $(x,\dot x)$ contained in $L^{-1}((0,\infty))$. As $(x,\dot x)$ is a boundary point of $\mathcal{T}$, $B$ intersects $\mathcal{T}$ and the complement $TM\setminus \mathcal{T}$. But then, the union $B \cup \mathcal{T}$ is connected, contained in $L^{-1}((0,\infty))$ and strictly contains $\mathcal{T}$, which is in contradiction with the maximal connectedness of $\mathcal{T}$ as a subset of $L^{-1}((0,\infty))$.

Second we find that set $\mathcal{S}:=\mathcal{T}_{x}\cap \{\dot{x}\in T_{x}M|L(x,\dot{x})\geq 1\}$ is convex, which follows from the fact that it is closed, connected and (strongly) locally convex, as we will show now.
\begin{enumerate}
	\item \textit{Proof that $\mathcal{S}$ is closed}:
	As the set $\{\dot{x}\in T_{x}M|L(x,\dot{x})\geq 1\}$ does not intersect $%
	\mathcal{N}=\ker L$ and $\partial \mathcal{T\subset N},$ it follows that we
	can as well write:%
	\begin{equation*}
	\mathcal{S}=\mathcal{\overline{T}}_x\cap \{\dot{x}\in T_{x}M|L(x,\dot{x})\geq 1\}.
	\end{equation*}%
	But, both the closure $\mathcal{\overline{T}}_x$ and the continuous preimage $\{%
	\dot{x}\in T_{x}M|L(x,\dot{x})\geq 1\}=L^{-1}([1,\infty ))$ are closed sets,
	i.e., $\mathcal{S}$ is closed.

	\item \textit{Proof that $\mathcal{S}$ is connected}:
	By the 2-homogeneity of $L$ and the fact that $\mathcal{T}_x$ is conic it is clear that we can write $\mathcal{T}_x$ as $\mathcal{T}_x = (0,\infty) \cdot \mathcal{S}$, i.e.\ every element $Y \in \mathcal{T}_x$ can be written as $y=\lambda z, z\in \mathcal{S}$ for some $\lambda\in (0,\infty)$. A straightforward argument shows that, if $\mathcal{S}$ would not be connected then $\mathcal{T}_x$ would not be connected. However $\mathcal{T}_x$ is connected by assumption and hence $\mathcal{S}$ must be connected.

	\item \textit{Proof that $\mathcal{S}$ is strongly locally convex}: $\mathcal{S}$ is locally strongly convex since the pull-back of $g^L$ is negative definite and $L>0$, similar to the arguments of Beem \cite[Lemma 1]{Beem}. As a remark, $\mathcal{T}_x$ contains an entire connected component of the indicatrix at $x$, i.e.\ of the set $\{\dot x\in T_xM | L(x,\dot x) = 1 \}$.
\end{enumerate}
Having established the closedness, connectedness and the strong local convexity of $\mathcal{S}$ it follows that $\mathcal{S}$ is convex, see~\cite{SSVLocalConvex} for the theorem. Eventually the convexity of $\mathcal{T}_x$ follows from its conicity and the convexity of $\mathcal{S}$.

\section{Randers, Kropina/Bogoslowski and m-th root Finsler spacetimes}\label{app:Ex}
Below our definition of Finsler spacetimes, Definition \ref{def:FST}, we claimed certain type of Finsler spacetimes are included by our definition. Here we provide some details how this can can be seen. An extensive analysis of these examples will be the subject of an upcoming publication.

Randers and Bogoslovsky/Kropina metrics are defined in terms of a Lorentzian metric $g$ and a $1$-form $A$ from which we can build the functions $A(\dot x) = A_a(x)\dot x^a$ and $g(\dot x, \dot x) = g_{ab}(x)\dot x^a \dot x^b$ on $TM$.

Let us start with Randers Finsler spacetimes whose Finsler Lagrange function is
\begin{align}
	L = \left(\sqrt{|g(\dot x, \dot x)|} + A(\dot x)\right)^2\,.
\end{align}
Calculating the corresponding $L$-metric yields
\begin{align}
	g^L_{ab} = \frac{g(\dot x, \dot x)}{|g(\dot x, \dot x)|} \left( g_{ab} \left(1 + \frac{A(\dot x)}{\sqrt{|g(\dot x, \dot x)|}} \right) + \frac{1}{\sqrt{|g(\dot x, \dot x)|}}(A_a \dot x_b  + A_b \dot x_a) \right)+ A_a A_b - \frac{A(\dot x)\ \dot x_a \dot x_b}{\sqrt{|g(\dot x, \dot x)|}^3}
\end{align}
and, calculating its determinant with help of the Levi-Civita $\varepsilon$-symbol\footnote{The determinant of a metric derived with help of the components of the totally antisymmetric Levi-Civita tensor density as $\varepsilon$-symbol $\det (g^L_{ab}) = \frac{1}{4!}\varepsilon^{a_1a_2a_3a_4}\varepsilon^{b_1b_2b_3b_4}g^L_{a_1b_1}g^L_{a_2b_2}g^L_{a_3b_3}g^L_{a_4b_4}$\,.} and the Mathematica tensor algebra extension xAct \cite{xact}, as shown in the included file \hyperlink{randers}{DetGLRanders.nb}, yields for $g$-timelike vectors ($g(\dot x,\dot x)>0$)
\begin{align}\label{eq:detgLR1}
	\det(g^L_{ab}) = \det(g_{ab})\frac{\left(\sqrt{g(\dot x, \dot x)} + A(\dot x)\right)^5}{\sqrt{g(\dot x, \dot x)}^5}\,.
\end{align}
This result coincides with the one derived in \cite[p. 284]{Shen} for the case of a positive definite metric $g$. For $g$-spacelike vectors ($g(\dot x,\dot x)<0$) the determinant becomes
\begin{align}\label{eq:detgLR2}
\det(g^L_{ab}) = \det(g_{ab})\left( 1 + \frac{5 A(\dot x)^2}{g(\dot x,\dot x)^2}(A(\dot x)^2 - 2 g(\dot x,\dot x)) + \frac{A(\dot x)}{\sqrt{-g(\dot x,\dot x)}^5}\left( ( A(\dot x)^2 - 5 g(\dot x, \dot x) )^2 - 20 g(\dot x, \dot x)^2\right)\right)\,.
\end{align}

For $0 < g^{-1}(A,A) < 1$, in particular $A$ being $g$-timelike and with bounded norm, $L=0$ yields a non-trivial null cone, which is the same as the null cone of the metric
\begin{align}
\tilde g(\dot x, \dot x) = g(\dot x, \dot x) - A(\dot x)A(\dot x) = 0\,.
\end{align}
For $A$ as above, this cone is sharper then the cone of $g$. Hence the connected component containing $A^\# = g^{-1}(A,\cdot)$ is bounded by $L=0$ and not intersected by the cone $g(\dot x, \dot x)=0$. On this connected component the signature of $g^L$ is identical to the one of $g$: the determinant of $g^L$, given by \eqref{eq:detgLR1} on the set of consideration, does not vanish and its eigenvalues can thus not change sign compared to the ones of $g$, since certainly for $\epsilon>A(\dot x)>0$ for any infinitesimal $\epsilon > 0$ the eigenvalues of $g$ and $g^L$ have identical signs. Hence we choose the connected component containing $A^\#$ as our future pointing timelike directions $\mathcal{T}_x$ for a Randers Finsler spacetime. Details about an analogue argument about the signature of the Finsler metric in Randers spaces, which are built from a positive definite metric, can be found in \cite[p.284]{Shen}.

Let us turn now to Bogoslovsky and Kropina Finsler spacetimes, which are special cases of Finsler Lagrangians of the type
\begin{align}\label{eq:ex2}
	L = |g(\dot x, \dot x)|^{1-q}A(\dot x)^{2q}\,.
\end{align}
Again we can calculate the determinant of the $L$-metric with help of the $\varepsilon$-symbol and the computer algebra program xAct, as shown in the included file \hyperlink{kropbog}{DetGLKropinaBogoslovsky.nb}. On the set for $g$-timelike vectors ($g(\dot x,\dot x)>0$) it is given by
\begin{align}
	\det(g^L_{ab}) = \det(g_{ab})(1-q)^3 \left[(1+q) A(\dot x)^2  - q g^{-1}(A,A) g(\dot x, \dot x) \right] g(\dot x, \dot x)^{-4q} (A(\dot x))^{-2(1-4q)}
\end{align}
The conditions on $A$ to yield a viable Finsler spacetime geometry depend on the value of $q$. For $q=0$ this expression is identical to $\det(g_{ab})$, for $q=1$ the $L$-metric is always degenerate. The more interesting cases are
\begin{enumerate}
	\item $q>1$: $\det(g^L_{ab})$ is negative if and only if $(1+q) A(\dot x)^2  - q g^{-1}(A,A) g(\dot x, \dot x) < 0$,
	\item $q=-1$: $\det(g^L_{ab})$ is negative if and only if $g^{-1}(A,A) g(\dot x, \dot x)>0$,
	\item $q<1$ and $q\neq-1,q\neq0$: $\det(g^L_{ab})$ is negative if and only if $(1+q) A(\dot x)^2  - q g^{-1}(A,A) g(\dot x, \dot x) > 0$.
\end{enumerate}
As in the Randers case, it depends on the causal character of $A$ on which connected subsets of $T_xM$ these conditions are satisfied. Now choose $A$ for instance to be $g$-timelike. In the $q>1$ case this implies that $\det g^L$ is always positive, since $(1+q) A(\dot x)^2  - q g^{-1}(A,A) g(\dot x, \dot x)>0$ by the reverse Cauchy-Schwartz inequality, and hence \eqref{eq:ex2} is not a Finsler spacetime. For $-1 < q < 1$, we find that $\mathcal{T}_x$ is the component where $g(\dot x, \dot x)>0$, i.e. the set of $g$-timelike vectors, which also satisfy $A(\dot x)\neq 0$. The signature of $g^L$ is again identical to the one of $g$, since it can not change as long as $\det(g^L_{ab})$ does not vanish. Without going into further details here we conclude that there exist Bogoslovsky and Kropina Finsler spacetimes for the right choice of~$A$.

Last but not least we mentioned the $m$-th root Finsler Lagrangians
\begin{align}
	L = |G_{a_1 \cdots a_m}(x)\dot x^{a_1}\ldots \dot x^{a_m}|^{\frac{2}{m}}\,.
\end{align}
A necessary condition for them to define a Finsler spacetime according to our definition is that the polynomial $G_{a_1 \cdots a_m}(x)\dot x^{a_1}\ldots\dot x^{a_m}$ is hyperbolic. Hyperbolic polynomials possess by definition so called hyperbolicity cones whose interior is the connected component $\mathcal{T}_x$. Hyperbolic polynomials are often discussed in the context of hyperbolic partial differential equations \cite[Sec.12.4]{Hoermander2} and the causal structure of physical field theories \cite{Raetzel:2010je}.

\section{Evaluation of the integrals $I_2$ and $I_3$}\label{app:I2I3}
In section \ref{sec:vari} we encountered the two integrals
\begin{align}
I_2 = \int_{D^+} \frac{1}{L} \frac{d\bar R}{dt}|_{t=0}\ dV^+_0\quad \textrm{ and }\quad I_3 =\int_{D^+} \frac{R}{L} \frac{1}{|\det g^L|}\frac{d |\det \bar g^{\bar L}|}{dt}|_{t=0}\ dV^+_0 \,,
\end{align}
which we will evaluate here in detail.

\subsection{The integral $I_2$}
The first step is to investigate the variation of the geodesic spray coefficients \eqref{eq:geodspray}, since they are the building blocks of the curvature scalar \eqref{eq:FinslerRicciS}. Denoting the derivatives of the variation by $v_i = v_{\cdot i}$ and $v_{ij} = v_{\cdot i \cdot j}$ we have
\begin{align}
	\bar L_{\cdot i} \overset{t^1}{\simeq} L_{\cdot i} + 2 tv_i,\quad \bar g^L_{ij} \overset{t^1}{\simeq} g^L_{ij} + t v_{ij},\quad \bar g^{Lij} \overset{t^1}{\simeq} g^{Lij} - t v^{ij}\,,
\end{align}
where $v^{ij} = g^{Lmi} g^{Lnj}v_{mn}$ and the symbol $\overset{t^1}{\simeq}$ means equality modulo higher than the first power in $t$. As a consequence the variation of the geodesic spray coefficients becomes
\begin{align}\label{eq:varG}
	2\bar G^i = \frac{1}{2}\bar g^{Lij}(\dot x^k \bar L_{\cdot j,k} - \bar L_{,j}) \overset{t^1}{\simeq} 2 G^i + t g^{Lij}(\dot x^k v_{j,k} - v_{,j} - 2 G^k v_{jk})
\end{align}
Since $\dot x$-differentiation preserves the tensor character, $v_i$ are covector components and it makes sense to speak about covariant derivatives thereof with respect to the Chern connection \eqref{eq:cherndef}: $v_i{}_{|j} = v_{i,j} - G^k{}_j v_{ik} - \Gamma^k{}_{ij}v_k$. Contracting the last index with $\dot x^j$ and taking into account the identities $\Gamma^k{}_{ij}\dot x^j = G^k{}_i$ and $G^k{}_j \dot x^j = 2 G^k$ yields
\begin{align}
	v_{i}{}_{,j}\dot x^j = v_i{}_{|j}\dot x^j + 2 G^k v_{ik} + G^k{}_i v_k,\quad \textrm{ and }\quad v_{,i} = v_{|i} + G^k{}_i v_k\,.
\end{align}
Substituting these equalities into \eqref{eq:varG} yields
\begin{align}
	2\bar G^i \overset{t^1}{\simeq} 2 G^i + 2 t A^i,\quad \textrm{ with }\quad A^i = \frac{1}{2}g^{Lij}(\nabla v_j - v_{|j})\,,
\end{align}
which agrees with the expression found in \cite{Chen-Shen}.

The second step is the variation of the Finsler Ricci scalar. By definition, we find the variation of the non-linear connection coefficients \eqref{eq:nldef} to be
\begin{align}
	\bar G^i{}_{j} \overset{t^1}{\simeq} G^i{}_j + t A^i{}_j,\quad \textrm{ where }\quad A^i{}_j = A^i{}_{\cdot j}\,.
\end{align}
Further, using the Landsberg tensor \eqref{eq:landsberg} $G^{i}{}_{jk} = G^i{}_{j\cdot k} = \Gamma^i{}_{jk} + P^i{}_{jk}$ and the definition of the Chern-Rund covariant derivative~\eqref{eq:cherndef}, we may write the variation of the non-linear curvature tensor \eqref{eq:nlcurv} as
\begin{align}
	\bar R^i{}_{jk}
	&= \bar{\delta}_k \bar{G}^i{}_j - \bar{\delta}_j \bar{G}^i{}_k\nonumber = \bar{G}^i{}_{j,k} - \bar{G}^i{}_{k,j} -  \bar{G}^l{}_k \bar{G}^i{}_{j\cdot l} + \bar{G}^l{}_j \bar{G}^i{}_{k\cdot l}\nonumber\\
	&\overset{t^1}{\simeq} G^i{}_{j,k} - G^i{}_{k,j}  - G^l{}_k G^i{}_{j\cdot l} + G^l{}_j G^i{}_{k\cdot l} + t( A^i{}_{j,k} - A^i{}_{k,j} -  A^l{}_kG^i{}_{j\cdot l} -  G^l{}_k A^i{}_{j\cdot l}  +  A^l{}_jG^i{}_{k\cdot l} +  G^l{}_j A^i{}_{k\cdot l})\nonumber\\
	&=R^i{}_{jk} + t (\delta_k A^i{}_j - \delta_j A^i{}_k -  A^l{}_kG^i{}_{j\cdot l} +  A^l{}_jG^i{}_{k\cdot l})\nonumber\\
	&=R^i{}_{jk} + t (A^i{}_{j|k} - A^i{}_{k|j} - \Gamma^i{}_{kl}A^l{}_j + \Gamma^i{}_{jl} A^l{}_k - A^l{}_kG^i{}_{j\cdot l} +  A^l{}_jG^i{}_{k\cdot l})\nonumber\\
	&=R^i{}_{jk} + t(A^i{}_{j|k}-A^i{}_{k|j} + A^l{}_j P^i{}_{kl} - A^l{}_k P^i{}_{jl})\,.
\end{align}
Contracting this equation with $\dot x^k$ and taking into account $P^i{}_{jk}\dot x^k = 0$ as well as $A^l{}_k \dot x^k = 2 A^l$ we get
\begin{align}
	\bar R^i{}_{j} \overset{t^1}{\simeq} R^i{}_{j} + t (\nabla A^i{}_j - 2 A^i{}_{|j} - 2 A^lP^i{}_{jl})\,,
\end{align}
which finally leads us to $\bar R \overset{t^1}{\simeq} R + t(\nabla A^i{}_i - 2 A^i{}_{|i} - 2 A^l P_l)$.

The third and final step is the isolation of the variation $2v$ in the integral. To do so we substitute our findings into the integral and use the identity \eqref{rel_integrals_SM_PTM} to equate
\begin{align}
	I_2 = \int_{D^+} \frac{1}{L} (\nabla A^i{}_i - 2 A^i{}_{|i} - 2 A^l P_l)\ dV^+_0\,.
\end{align}
Observe that the first term gives a boundary term which we can neglect. This is so since $\nabla L = 0$ and thus $\nabla(\frac{A^i{}_i}{L})$ is a divergence of a $0$-homogeneous vector field according to \eqref{eq:difnabla}. It remains to write $L^{-1}A^i{}_{|i} = \mathrm{div}(L^{-1}A^i\delta_i) + L^{-1}A^iP_i$, see \eqref{divergence_horizontal}, to find
\begin{align}
	I_2 = - \int_{D^+} \frac{4}{L} A^l P_l\ dV^+_0\,.
\end{align}
Using the definition of $A^i$ we can expand the integrand as
\begin{align}
	- 4L^{-1}A^iP_i = 2\big[(L^{-1}vP^i)_{|i} - L^{-1}vP^i{}_{|i} - \nabla(L^{-1}v_iP^i) + L^{-1}v_i\nabla P^i\big]\,,
\end{align}
and observe that $\nabla(L^{-1}v_iP^i)$ is a total divergence of a $0$-homogeneous vector field again and that $(L^{-1}vP^i)_{|i} = \mathrm{div}(L^{-1}vP^i\delta_i) + L^{-1}vP^iP_i$, which implies
\begin{align}
	- 4L^{-1}A^iP_i =\mathrm{div}(\ldots) + 2 L^{-1}( P^i P_i - P^i{}_{|i})v + 2 L^{-1}v_m \nabla P^m\,.
\end{align}
Hence, the last we need to investigate is $2 L^{-1}v_m \nabla P^m$. By the Leibniz rule we have
\begin{align}
	L^{-1}v_m \nabla P^m = (L^{-1}v\nabla P^i)_{\cdot i} - L^{-1}{}_{\cdot i}v\nabla P^i - L^{-1}v(\nabla P^i)_{\cdot i}\,.
\end{align}
The second term on the right hand side vanishes since $\dot x_i\nabla P^i = \nabla(\dot x^i P_i) = 0$. The first term can be written into a divergence of a vertical vector field according to \eqref{divergence_vertical} as
$(L^{-1}v\nabla P^i)_{\cdot i} = \mathrm{div}(L^{-1}v\nabla P^i\dot{\partial}_i) - 2 C_i L^{-1}v\nabla P^i + 0$. Summing up yields
\begin{align}
	L^{-1}v_m \nabla P^m = \mathrm{div}(\ldots) - v L^{-1}(2 C_i \nabla P^i + (\nabla P^i)_{\cdot i})
\end{align}
It is now straightforward to see that $2 C_i \nabla P^i + (\nabla P^i)_{\cdot i} = g^{Lij}(\nabla P_i)_{\cdot j}$ and so altogether
\begin{align}
	- 4L^{-1}A^iP_i = \mathrm{div}(\ldots) + 2 L^{-1}( P^i P_i - P^i{}_{|i} - g^{Lij}(\nabla P_i)_{\cdot j})v \,.
\end{align}

Finally the integral $I_2$ becomes, neglecting the boundary terms,
\begin{align}
	I_2 = \int_{D^+} \frac{2}{L}( P^i P_i - P^i{}_{|i} - g^{Lij}(\nabla P_i)_{\cdot j})v\ dV^+_0\,,
\end{align}
which is the desired expression \eqref{eq:int2}.

\subsection{The integral $I_3$}
For the integral $I_3$ observe that
\begin{align}
	|\det \bar g^L| \overset{t^1}{\simeq} |\det g^L| + t g^{Lij} v_{ij} |\det g^L|\,,
\end{align}
where we used the derivative formula for the determinant
\begin{align}
	\frac{d}{dt}\det \bar g^L = \det \bar g^L \bar g^{Lij}\frac{d}{dt} \bar g^L_{ij}\,.
\end{align}
The integral thus becomes
\begin{align}
	I_3 =\int_{D^+} \frac{R}{L} g^{Lij} v_{ij}\ dV^+_0
\end{align}
The Leibniz rule, together with \eqref{divergence_vertical} imply
\begin{align}
	 L^{-1} R g^{Lij} v_{ij}
	 &= \mathrm{div}( L^{-1} R g^{Lij}v_j \dot{\partial}_i) - 2 L^{-1} R C^i v_i + 8 L^{-2} R v - (L^{-1} R g^{Lij})_{\cdot i}v_j\\
	 &= \mathrm{div}( L^{-1} R g^{Lij}v_j \dot{\partial}_i) + 12 L^{-2} R v - L^{-1} R_{\cdot i} v^i\\
	 &= \mathrm{div}( L^{-1} R g^{Lij}v_j \dot{\partial}_i) - \mathrm{div}( L^{-1} R_{\cdot{} j} g^{Lij}v \dot{\partial}_i) + L^{-1}v g^{Lij}R_{\cdot i\cdot j}\,.
\end{align}
Hence the integral turns out to be, neglecting again the boundary terms,
\begin{align}
	I_3 =\int_{D^+} \frac{1}{L}g^{Lij}R_{\cdot i\cdot j}v\ dV^+_0\,,
\end{align}
which again is the result presented in \eqref{eq:int3}.

\section{Proof of Lemma \ref{lem:intbdry}}\label{app:Lemmaint}
In Lemma \ref{lem:intbdry} we displayed useful formulae to understand the Finsler gravity action in the case of pseudo-Riemannian geometry. Here we provide the proof of the equations \eqref{eq:L3e1}, \eqref{eq:L3e2} and \eqref{eq:gxx}.

For the first equation we expand
\begin{align}
	g^{Lij} (Lf)_{\cdot i\cdot j}dV^+_0 = g^{Lij}(2 g^L_{ij} f +2 L_{\cdot i}f_{\cdot j} + L f_{\cdot i \cdot j})dV^+_0\,.
\end{align}
The $0$-homogeneity of $f$ with respect to $\dot x$ and $\textrm{dim}(M)=4$ implies
\begin{align}
	g^{Lij} (Lf)_{\cdot i\cdot j}dV^+_0 = 8 f dV^+_0 + L g^{Lij} f_{\cdot i \cdot j}dV^+_0\,.
\end{align}
The last step is to show that the last term in this sum is a total derivative term. This can be seen from
\begin{align}
	L g^{Lij} f_{\cdot i \cdot j}dV^+_0 = L g^{Lij} f_{\cdot i \cdot j} \tfrac{|\det g^L|}{L^2}\mathbf{i}_\mathbb{C}\textrm{Vol} = (g^{Lij} f_{\cdot i} \tfrac{|\det g^L|}{L})_{\cdot j}\mathbf{i}_\mathbb{C}\textrm{Vol} = d(i_X dV^+_{0})\,,
\end{align}
where $X = L g^{Lij} f_{\cdot i}\dot \partial_i$.

The second equation can be proven  by expanding
\begin{align}
	L^{-1}g^{Lij}(L^2 f)_{\cdot i\cdot j} = L^{-1} g^{Lij}(2 g^L_{ij}(Lf) + 2 L_{\cdot i}(Lf)_{\cdot j} + L (Lf)_{\cdot i \cdot j}) = 16 f + g^{Lij}(Lf)_{\cdot i \cdot j}\,.
\end{align}
Using the first equation in the last term we obtain the desired result.

\bibliographystyle{utphys}
\bibliography{VC}

\newpage
\hypertarget{randers}{}
\includepdf[pages=1]{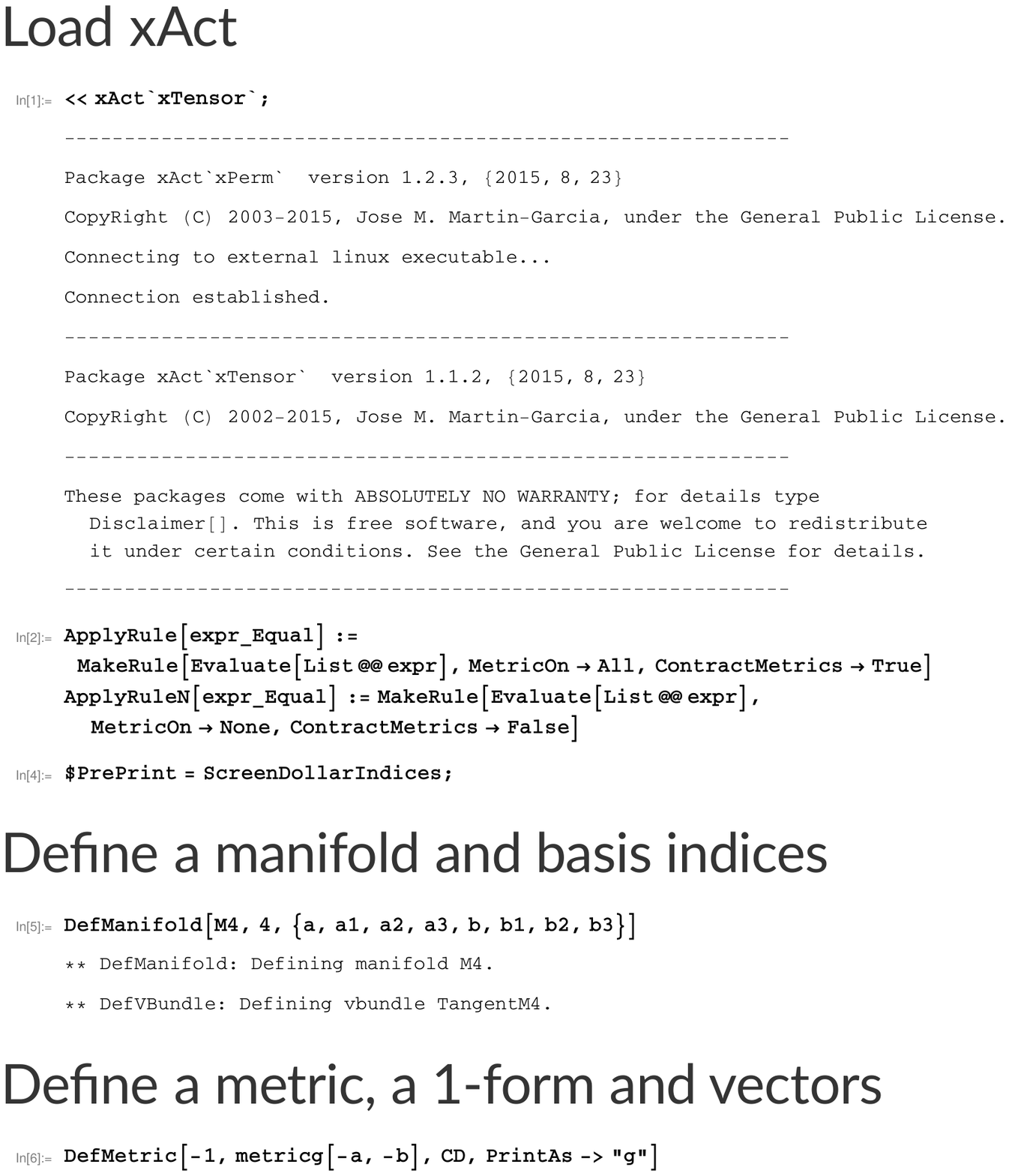}
\includepdf[pages=2]{xActRanders.pdf}
\includepdf[pages=3]{xActRanders.pdf}
\includepdf[pages=4]{xActRanders.pdf}
\includepdf[pages=5]{xActRanders.pdf}
\includepdf[pages=6]{xActRanders.pdf}
\hypertarget{kropbog}{}
\includepdf[pages=1]{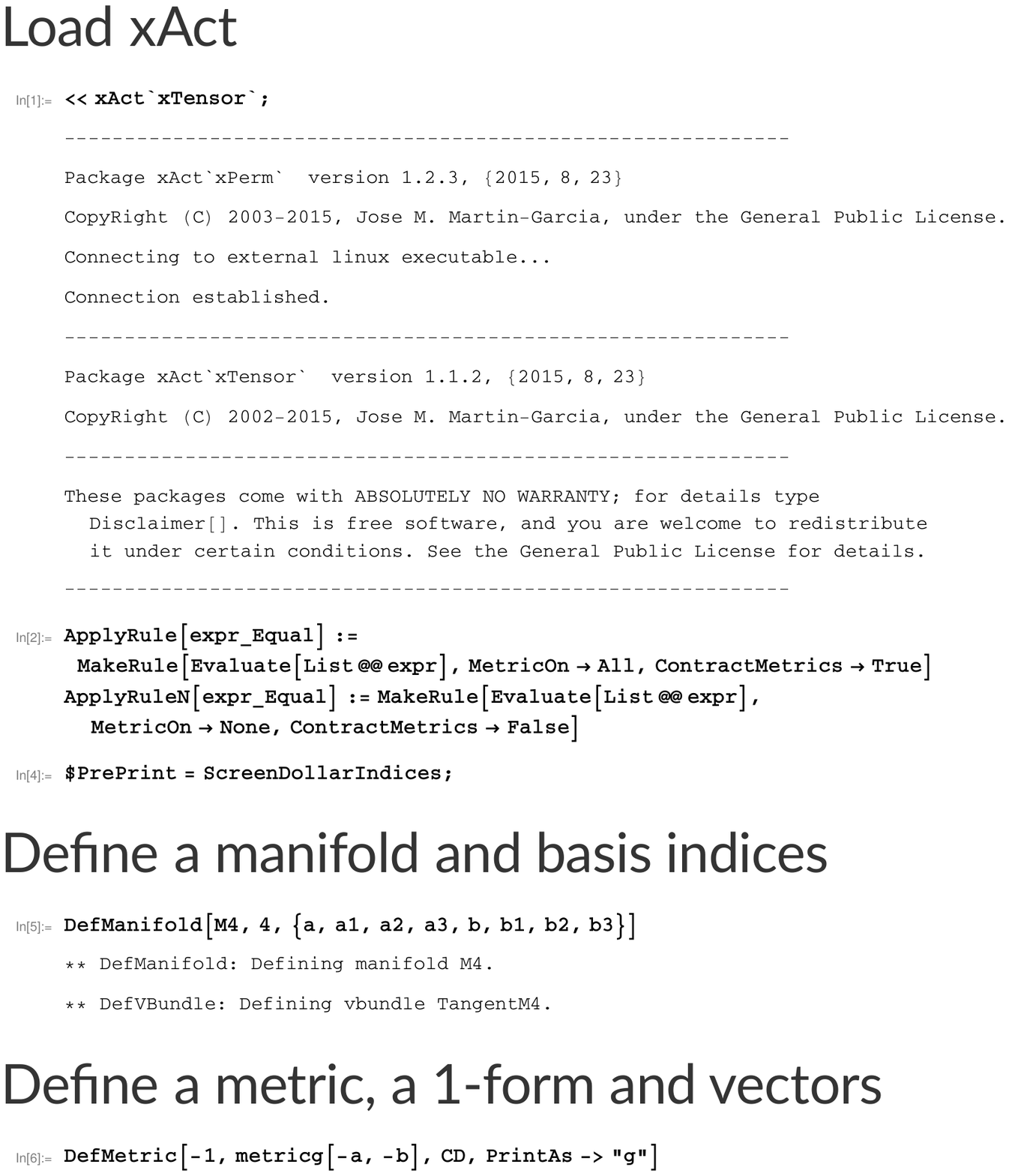}
\includepdf[pages=2]{xActKropBog.pdf}
\includepdf[pages=3]{xActKropBog.pdf}
\includepdf[pages=4]{xActKropBog.pdf}
\includepdf[pages=5]{xActKropBog.pdf}

\end{document}